\documentclass[journal]{IEEEtran}

\usepackage{cite}
\usepackage{gensymb}
\usepackage{amsmath,amssymb,amsfonts}
\usepackage{bm}
\usepackage{amsthm}

\usepackage{algorithm}
\usepackage{algpseudocode}

\usepackage{array}

\usepackage{graphicx}
\usepackage{textcomp}
\usepackage{xcolor}
\usepackage{comment}
\usepackage{svg}
\usepackage{subcaption}
\usepackage[margin=16.7mm,top=16.8mm,bottom=16.8mm]{geometry}
\usepackage[normalem]{ulem} 
\usepackage{cancel}
\usepackage{mathtools, cuted}

\usepackage{tikz}
\usetikzlibrary{shapes,arrows,positioning}
\usetikzlibrary{3d}

\newcommand{\com}[1]{\textcolor{red}{#1}}
\newcommand{\comm}[1]{\textcolor{blue}{#1}}

\newtheorem{Lemma}{Lemma}

\newtheorem{Corollary}{Corollary}

  {\proof}{\proofend}
\newtheorem{proposition}{Proposition}

\newcommand{\qg}{{\bf g}}

\newcommand{\qn}{{\bf n}}

\newcommand{\qu}{{\bf u}}
\newcommand{\qv}{{\bf v}}
\newcommand{\qw}{{\bf w}}

\newcommand{\qA}{{\bf A}}
\newcommand{\qB}{{\bf B}}
\newcommand{\qC}{{\bf C}}

\newcommand{\qI}{{\bf I}}

\newcommand{\qN}{{\bf N}}

\newcommand{\qW}{{\bf W}}

\newcommand{\qY}{{\bf Y}}

\newcommand{\diag}{\mathrm{diag}}
\newcommand{\trace}{\mathrm{tr}}

\newcommand{\KE}{\mathcal{K}}
\newcommand{\LA}{\mathcal{L}}

\newcommand{\gkl}{{\textbf{g}_{kl}}} 
\newcommand{\gklt}{{\textbf{g}^T_{kl}}} 
\newcommand{\gklc}{{\textbf{g}^*_{kl}}} 
\newcommand{\gklcp}{{\textbf{g}^*_{kl'}}} 

\newcommand{\gamkl}{{{\gamma}_{kl}}}
\newcommand{\gamil}{{{\gamma}_{il}}}
\newcommand{\gamklp}{{{\gamma}_{kl'}}}

\newcommand{\gamklsq}{{{\gamma}_{kl}^2}}

\newcommand{\gamklcb}{{{\gamma}_{kl}^3}}

\newcommand{\gamkltt}{{{\gamma}_{kl}^4}}

\newcommand{\aik}{{\alpha_{ik}}}
\newcommand{\aiksq}{{\alpha_{ik}^2}}
\newcommand{\aiktt}{{\alpha_{ik}^4}}
\newcommand{\oikl}{{{\varrho}_{ik,l}}} 
\newcommand{\oiklp}{{{\varrho}_{ik,l'}}}

\newcommand{\oiklsq}{{{\varrho}^2_{ik,l}}} 
\newcommand{\oiklsqp}{{{\varrho}^2_{ik,l'}}}

\newcommand{\kkl}{{K_{kl}}}

\newcommand{\kmrt}{{\kappa_{kl}}}
\newcommand{\kmrti}{{\kappa_{il}}}
\newcommand{\kmrtip}{{\kappa_{il'}}}

\newcommand{\ckl}{{c_{kl}}}

\newcommand{\bkil}{{\zeta_{ik,l}}}
\newcommand{\bkilsq}{{\zeta^2_{ik,l}}}
\newcommand{\bkilp}{{\zeta_{ik,l'}}}
\newcommand{\bkilpsq}{{\zeta^2_{ik,l'}}}

\newcommand{\bklqp}{{\varsigma_{kl'}}}
\newcommand{\bilqp}{{\varsigma_{il'}}}

\newcommand{\bklq}{{\varsigma_{kl}}}
\newcommand{\bilq}{{\varsigma_{il}}}
\newcommand{\bklqsq}{{\varsigma^2_{kl}}}
\newcommand{\bilqsq}{{\varsigma^2_{il}}}

\newcommand{\gil}{{\textbf{g}_{il}}} 

\newcommand{\tauP}{{\tau_p p_p}}

\newcommand{\wkl}{{\qw_{kl}}} 

\newcommand{\Ex}{\mathbb{E}}
\newcommand{\NL}{\mathtt{NL}}

\newcommand{\hgkl}{{\hat{\textbf{g}}_{kl}}}

\newcommand{\thgkl}{{\Tilde{\hat{\textbf{g}}}_{kl}}}

\newcommand{\thgklc}{{\Tilde{\hat{\textbf{g}}}^*_{kl}}}
\newcommand{\thgklt}{{\Tilde{\hat{\textbf{g}}}^T_{kl}}}

\newcommand{\tekl}{{\Tilde{\hat{\pmb{\varepsilon}}}_{kl}}}
\newcommand{\teklt}{{\Tilde{\hat{\pmb{\varepsilon}}}^T_{kl}}}
\newcommand{\teklc}{{\Tilde{\hat{\pmb{\varepsilon}}}^*_{kl}}}

\newcommand{\thgklcp}{{\Tilde{\hat{\textbf{g}}}^*_{kl'}}}
\newcommand{\thgkltp}{{\Tilde{\hat{\textbf{g}}}^T_{kl'}}}

\newcommand{\thgil}{{\Tilde{\hat{\textbf{g}}}_{il}}}

\newcommand{\thgilc}{{\Tilde{\hat{\textbf{g}}}^*_{il}}}
\newcommand{\thgilt}{{\Tilde{\hat{\textbf{g}}}^T_{il}}}

\newcommand{\hgil}{{\hat{\textbf{g}}_{il}}}

\newcommand{\hgilt}{{\hat{\textbf{g}}^T_{il}}}
\newcommand{\hgilc}{{\hat{\textbf{g}}^*_{il}}}
\newcommand{\hgiltp}{{\hat{\textbf{g}}^T_{il'}}}

\newcommand{\bgkl}{{\Bar{\textbf{g}}_{kl}}} 
\newcommand{\bgklc}{{\Bar{\textbf{g}}^*_{kl}}} 
\newcommand{\bgklt}{{\Bar{\textbf{g}}^T_{kl}}} 

\newcommand{\bgilc}{{\Bar{\textbf{g}}^*_{il}}} 

\newcommand{\bgklcp}{{\Bar{\textbf{g}}^*_{kl'}}} 

\newcommand{\bhkl}
{{\Bar{\textbf{h}}_{kl}}} 
\newcommand{\bhklt}
{{\Bar{\textbf{h}}^T_{kl}}} 
\newcommand{\bhklc}
{{\Bar{\textbf{h}}^*_{kl}}} 
\newcommand{\bhklcp}
{{\Bar{\textbf{h}}^*_{kl'}}} 

\newcommand{\bhilc}
{{\Bar{\textbf{h}}^*_{il}}} 

\newcommand{\bhil}
{{\Bar{\textbf{h}}_{il}}} 
\newcommand{\bhilt}
{{\Bar{\textbf{h}}^T_{il}}} 

\newcommand{\bgil}{{\Bar{\textbf{g}}_{il}}} 
\newcommand{\bgilt}{{\Bar{\textbf{g}}^T_{il}}} 

\newcommand{\bgiltp}{{\Bar{\textbf{g}}^T_{il'}}} 

\newcommand{\tgkl}{{\Tilde{\textbf{g}}_{kl}}} 

\newcommand{\tgil}{{\Tilde{\textbf{g}}_{il}}} 

\newcommand{\tgklc}{{\Tilde{\textbf{g}}^*_{kl}}} 
\newcommand{\tgklt}{{\Tilde{\textbf{g}}^T_{kl}}} 

\newcommand{\tgklcp}{{\Tilde{\textbf{g}}^*_{kl'}}} 
\newcommand{\tgkltp}{{\Tilde{\textbf{g}}^T_{kl'}}} 

\newcommand{\zekl}{{\xi_{kl}}} 
\newcommand{\bkl}{{\beta_{kl}}} 
\newcommand{\bil}{{\beta_{il}}} 
\newcommand{\bklsq}{{\beta^2_{kl}}} 
\newcommand{\bilsq}{{\beta^2_{il}}} 

\newcommand{\bklp}{{\beta_{kl'}}} 

\newcommand{\zkl}{{\textbf{z}_{kl}}} 
\newcommand{\In}{{\textbf{I}_N}} 
\newcommand{\sui}{{\sum\nolimits_{i \in \mathcal{P}_k}}} 
\newcommand{\suk}{{\sum\nolimits_{k\in\KE}}}
\newcommand{\suik}{{\sum\nolimits_{i\in\KE}}}

\newcommand{\sul}{{\sum\nolimits_{l\in\LA}}}

\newcommand{\slpL}{{\sum\nolimits_{l'\in\LA\setminus l}}}
\newcommand{\slp}{{\sum\nolimits_{l'\in\LA}}}

\newcommand{\xEl}{{\textbf{x}^E_{l}}} 

\newcommand{\rEk}{{r^E_{k}}} 
\newcommand{\IEk}{{{I}^E_{k}}} 
\newcommand{\EEk}{{{E}^{\NL}_{k}}} 
\newcommand{\ECk}{{{E}^C_{k}}} 
\newcommand{\dEk}{{\Delta {E}_k}} 

\newcommand{\GMp}{{\upsilon^{(q)}_{i,l,l'}}} 

\newcommand{\op}{{\Omega^{(q)}}}

\newcommand{\opkl}{{\Omega^{(q)}_{kl}}}

\newcommand{\opklv}{{\breve{\Omega}^{(q)}_{kl}}}

\newcommand{\okl}{{\Omega_{kl}}}

\newcommand{\opil}{{\Omega^{(q)}_{il}}}
\newcommand{\opild}{{\Omega^{(q)}_{il'}}}

\newcommand{\bx}{{\bm{x}^{(q)}_{\bar{k}}}}
\newcommand{\bxo}{{\bm{x}^{(q)}_{{\bar{k}},0}}}

\newcommand{\by}{{\bm{y}^{(q)}_{\bar{k}}}}
\newcommand{\byx}{{\bm{y}^{(q)}_{{\bar{k}},x}}}
\newcommand{\byu}{{\bm{y}^{(q)}_{{\bar{k}},\mu}}}

\newcommand{\bgx}{{\bm{g}\!\left(\bx\right)}}
\newcommand{\fx}{{f\!\left(\bx\right)}}
\newcommand{\bu}{{\bm{\mu}}}
\newcommand{\bl}{{\bm{\lambda}}}

\newcommand{\ME}{{\mathbb{E}}} 
\newcommand{\MV}{{\mathrm{Var}}} 

\def\BibTeX{{\rm B\kern-.05em{\sc i\kern-.025em b}\kern-.08em
    T\kern-.1667em\lower.7ex\hbox{E}\kern-.125emX}}

\begin{document}

\title{Optimal Energy Harvesting Strategy in Cell-Free Massive MIMO Using Markov Process Evolution}

\author{\IEEEauthorblockN{Muhammad Zeeshan Mumtaz {\href{https://orcid.org/0000-0002-0718-8822}{\includesvg[scale=0.06]{orcid_id.svg}}}, \textit{Graduate Student Member, IEEE}, Mohammadali Mohammadi {\href{https://orcid.org/0000-0002-6942-988X}{\includesvg[scale=0.06]{orcid_id.svg}}}, \textit{Senior Member, IEEE}, Hien Quoc Ngo {\href{https://orcid.org/0000-0002-3367-2220}{\includesvg[scale=0.06]{orcid_id.svg}}}, \textit{Senior Member, IEEE}, and Michail Matthaiou {\href{https://orcid.org/0000-0001-9235-7741}{\includesvg[scale=0.06]{orcid_id.svg}}}}, \textit{Fellow, IEEE}}

\title{ Optimized Energy Harvesting in Cell-Free Massive MIMO Using Markov Process Evolution}

\author{Muhammad Zeeshan Mumtaz,~\IEEEmembership{Student Member,~IEEE,} Mohammadali Mohammadi,~\IEEEmembership{Senior Member,~IEEE,} \\
Hien Quoc Ngo,~\IEEEmembership{Fellow,~IEEE,}  and  Michail Matthaiou,~\IEEEmembership{Fellow,~IEEE}
\\
\thanks{

This work was supported by the U.K. Engineering and Physical Sciences Research Council (EPSRC) (grant No. EP/X04047X/1). The work of  H. Q. Ngo
 was supported by the U.K. Research and Innovation Future
Leaders Fellowships under Grant MR/X010635/1, and a research grant from the Department for the Economy Northern Ireland under the US-Ireland R\&D Partnership Programme. The work of M. Z. Mumtaz, M. Mohammadi and M. Matthaiou was supported by the European
Research Council (ERC) under the European Union’s Horizon 2020 research
and innovation programme (grant agreement No. 101001331).}
\thanks{The authors are with the Centre for Wireless Innovation (CWI), Queen's University Belfast, BT3 9DT Belfast, U.K., (email:\{mmumtaz01, m.mohammadi, hien.ngo, m.matthaiou\}@qub.ac.uk).}
\thanks{M. Z. Mumtaz is also with the College of Aeronautical Engineering, National University of Sciences \& Technology (NUST), Pakistan, (email: zmumtaz@cae.nust.edu.pk).}
\thanks{ Parts of this paper have been presented at the 2024 IEEE GLOBECOM~\cite{Zeeshan:GC:2024}}.
}

\maketitle

\begin{abstract}
This paper investigates a discrete energy state transition model for energy harvesting (EH) in cell-free massive multiple-input-multiple-output (CF-mMIMO) networks. A Markov chain-based stochastic process is conceived to characterize the temporal evolution of the user equipment (UE) energy level by leveraging state transition probabilities (STP) based on the energy differential ($\Delta E$) between the EH and consumed energy within each coherence interval. Tractable mathematical relationships are derived for the STP cases using a new stochastic model of non-linear EH, approximated using a Gamma distribution. This derivation leverages closed-form expressions for the mean and variance of the harvested energy. To improve the positive STP of the minimum energy UE among all network UEs, we aim to maximize the $\Delta E$ for this UE using two power allocation (PA) schemes. The first scheme is a heuristic PA using the relative channel characteristics to this UE from all access points (APs). The second scheme is the optimized PA based on the solution of a second-order conic problem to maximize the $\Delta E$ using a responsive primal-dual interior point method (PD-IPM) algorithm with modified backtracking line-search, iterating over multiple PA periods. Our simulation results illustrate that both the proposed PA schemes enhance the dynamic minimum UE energy level by around four-fold over full power control, along with the performance improvement attributed to spatial resource diversification of CF-mMIMO systems. 
\end{abstract}


\IEEEpeerreviewmaketitle

\vspace{-1em}
\section{Introduction}
Wireless power transfer (WPT) stands out as a pivotal technology for the realization of future 6G networks. These networks are designed to support a plethora of applications related to the Internet of Things (IoT), facilitating massive connectivity among numerous small-scale devices~\cite{matthaiou,pengcheng,Zheng}. The development of such systems hinges on the availability of sustainable energy sources. In this context, WPT plays a key role by enabling the harnessing of energy from ubiquitous radio-frequency (RF) signals. These signals, broadcasted by ambient or dedicated wireless transmitters, permeate our surroundings, offering a promising solution for powering the extensive network of interconnected devices, where it is impractical to provide electric grid supply to large number of nodes~\cite{Clerckx:JSAC:2019}. Thus, these network nodes have to rely on electromagnetic (EM) waves from power source transmitters to receive and harvest electrical energy to underpin the device operation and information transfer. Wireless sensor networks, unmanned aerial vehicle based communications, and human body data acquisition sensors, are a few applications whose evolution depends upon the pragmatic design and development of wireless energy harvesting (EH) technologies~\cite{hujie,babbar}. Despite its potential, WPT faces significant challenges, notably the rapid decline of energy transfer efficiency with increasing distance, a phenomenon attributed to the distance-dependent path loss. This limitation along with the ineffective EM beam focusing techniques underscore the necessity for innovative solutions to enhance the viability of WPT in future network architectures~\cite{Lu:Tut:2015}. 

For the effective deployment of WPT in wireless networks, cell-free massive multiple-input-multiple-output (CF-mMIMO) presents itself as a disruptive approach. Within the CF-mMIMO framework, the spatial distance between the APs and UEs is considerably reduced compared to traditional cellular network structures. This diminished distance between the APs and UEs significantly alleviates the impact of path loss attenuation that typically impede the efficiency of remote charging mechanisms~\cite{Hienenergy,Hien:PIEEE:2024,van,Mohammadi}. Furthermore, the incessant evolution in CF-mMIMO hardware technologies contributes to the enhancement of energy transmission capabilities. These advancements facilitate the more accurate alignment of EM beams from the APs towards the UEs, while simultaneously minimizing side-lobe levels. Such improvements are instrumental in augmenting the efficacy of energy reception at the UEs' antennas, thereby optimizing the overall performance of wireless power delivery systems. This technological progress underscores the potential of CF-mMIMO to serve as a foundational element in the development and implementation of future-oriented WPT industrial solutions.

\vspace{-0.5em}
\subsection{Review of Related Literature}
Numerous research works have demonstrated the effectiveness of CF-mMIMO architecture in the domain of EH of IoT devices, while devising strategies for device storage charging with more emphasis on network spectral efficacy than energy efficiency~\cite{Wang:JIOT:2020,Demir,femenias,loku, Zhang:IoT:2022, mohammadi2, yaozhang}. In \cite{Wang:JIOT:2020}, a small-scale wireless powered IoT based CF-mMIMO system was proposed. A performance comparison between co-located and CF-mMIMO was presented with the verdict of superiority of the latter technique, along with the minimization of the total transmit energy consumption with target signal-to-noise-ratio (SNR) constraints.  Similar wireless-powered CF-mMIMO systems have been considered by Demir \textit{et al.} in \cite{Demir}, whose goal was to maximize the minimum spectral efficiency for power control coefficients and large-scale fading decoding vectors.  In \cite{femenias}, a simultaneous wireless information and power transfer (SWIPT)-enhanced framework was presented for CF-mMIMO networks with both EH mobile stations (EMS) requiring WPT  and conventional mobile stations (MS), not requiring WPT. The trade-off between the achievable spectral efficiency at all MSs and instantaneous energy efficiency of EMSs was analyzed in relation to various communication parameters, such as pilot training  length, EH phase length and pilot transmit power. Another identical CF-mMIMO network for enabling SWIPT was contemplated by the authors of \cite{loku}, who explored the joint optimization of harvested energy and target rate for both time switching and power splitting protocols. In \cite{Zhang:IoT:2022}, a CF-mMIMO IoT network with SWIPT was considered with accelerated projected gradient method-based max-min power control strategy for the average harvested energy and achievable rate separately. It provided a runtime comparison against a convex-solver based method presented in~\cite{hien}, however, the solution accuracy was not deliberated upon. In~\cite{ mohammadi2}, a novel CF-mMIMO structure was proposed to simultaneously support information UEs and energy UEs (IoT devices). Recently in~\cite{yaozhang}, a CF-mMIMO aided SWIPT-IoT network topology was presented with a practical approach of using non-ideal, low-resolution analog-to-digital \& digital-to-analog converters, whose sum-EH and sum-rate performance was enhanced by particle swarm optimization framework. 

\vspace{-0.5em}
\subsection{Research Gap and Main Contributions}
Although the above research works signify the performance enhancement in WPT applications of CF-mMIMO systems, the dynamics of device energy states and their transitions over multiple coherence intervals have been largely overlooked, particularly in relation to the energy differential between energy harvested and energy consumed within a certain time interval. These important factors are crucial for determining the optimal EH strategy within a WPT-IoT network, as their profound understanding along with the device storage status, will allow devising more efficient and adaptive energy management protocols that would enhance the network's operational longevity and reliability. In order to encompass the transitions of user energy states in CF-mMIMO systems, a Markov chain (MC)-based stochastic process should be considered which provides a statistical framework that models a system's evolution over distinct states \cite{meyn}. Although MC processes have been extensively used for understanding the transition of energy states in EH-based communication networks, the contemporary literature on MC application for WPT in CF-mMIMO or MIMO systems is still limited. Kusaladharma \textit{et al.} introduced the concept of MC based energy state transitions in SWIPT based CF-mMIMO structure, but confined their analysis to the probability of fully charged state~\cite{kusal}. Moreover, this work lacks a well-defined PA scheme to enhance the EH performance over the evolution of the Markov process. Conversely, existing studies on MC applications in EH/WPT-based networks \cite{Iwaki,Salehi,Moon_2,fangchao,teng} primarily focus on energy storage state transitions within battery electronics, emphasizing strategies, such as harvest-store-use (HSU) and harvest-then-access, without providing a comprehensive statistical analysis of the harvested power.

Motivated by the above discussion, our research focuses on developing a CF-mMIMO based EH mechanism that employs a Markov state transition process, addressing the neglected aspects of energy state dynamics and consumption in the literature~\cite{kusal}, thereby offering a more robust framework for optimizing the EH schemes in wireless networks. The main contributions of our research are as follows: 

\begin{itemize}
    \item We propose a MC-based state transition process of discrete energy levels for IoT device energy storage, which depends on EH from the APs for network operation. Moreover, we derive mathematical relationships for state transition probabilities as a function of the energy differential after power consumption during uplink (UL) training and data transfer phases.
    \item An analytical solution is derived to provide closed-form expressions for the statistical parameters, specifically mean and variance, of non-linear EH in the proposed CF-mMIMO system. These characteristics are also utilized to approximate the energy accumulation in UE battery storage using a Gamma distribution. Numerical results substantiate the accuracy and utility of this approximation for further analysis of state transition probabilities.
    \item We propose two power allocation (PA) schemes for maximization of the energy differential for the minimum energy level UE among all the network UEs at the end of the previous PA period, by invoking the Markov property of single state dependency. Firstly, a heuristic channel characteristics based PA scheme (CCPA) is designed based on the relative channel gains from all APs to this particular UE. Secondly, we propose an optimized PA scheme addressing the problem of maximizing the energy differential.  The original concave objective function is further transformed into a second-order conic (SOC) problem bounded by power coefficient constraints.  
    \item A responsive primal-dual interior point method (PD-IPM) based optimization algorithm is proposed to solve the SOC problem. This framework converges towards the optimal point by iteratively updating the search direction using modified backtracking line-search, while satisfying modified Karush-Kuhn-Tucker (KKT) conditions.
    \item Our simulation results clearly demonstrate the dynamic capability of both the proposed PA architectures, which significantly enhances the minimum device energy level of the CF-mMIMO network over the  Markov process. Moreover, it is shown that more diverse spatial distribution of service antennas leads to remarkable improvement in the EH performance of the proposed CF-mMIMO architecture. 
\end{itemize}

\subsection{Paper Organization and Notation}
The rest of this paper is organized as follows: In Section \ref{sec:system_model}, we present the system model, deliberating upon the CF-mMIMO framework along with channel characterization. Section \ref{sec:DL_EH} delves into the downlink (DL) EH mechanism, including the derivation of statistical parameters of the harvested energy. In Section \ref{sec:Markov}, we develop a Markov process-based model for probabilistic energy state transitions. Section \ref{sec:optimization} formulates two PA schemes for maximizing the minimum UE energy levels, while discussing a heuristic technique based on the relative channel characterization. This section also introduces iterative application of the proposed responsive PD-IPM algorithm for optimized PA scheme. Finally, Section \ref{sec:results} provides numerical results and discusses the performance of the both PA frameworks, followed by important conclusions in Section \ref{sec:conclusion}.

\emph{Notations:} We use bold upper-case letters to denote matrices, and bold lower-case letters to denote vectors. The superscripts $(\cdot)^{\rm{H}}$ and $(\cdot)^{\rm{T}}$ denote the Hermitian transpose and transpose of a matrix, respectively; $\qI_M$ represents the $M\times M$ identity matrix;
$\| \cdot \|$ returns the norm of a matrix; $\mathbb{E}\{\cdot\}$ and $\MV\{\cdot\}$ denote the statistical expectation and variance, respectively; $\gamma(a,x)$ denotes the lower incomplete Gamma function~\cite[Eq. (8.350.1)]{Integral:Series:Ryzhik:1992}; $\Gamma(a)$ is the Gamma function~\cite[Eq. (8.310)]{Integral:Series:Ryzhik:1992}. Finally, $\delta_{l,l'}$ is the Kronecker delta function.

\section{System Model}
\label{sec:system_model}

We consider a CF-mMIMO network operating under time-division duplex mode, where $L$ APs, each equipped with $N$ antennas ($N \geq 1$), offer communication services including WPT to $K$ UEs, each having a single antenna. For the purposes of clarity and ease of reference, the indices of UEs and APs are denoted by the sets $\KE\triangleq \{1,\ldots,K\}$ and $\LA\triangleq \{1,\ldots,L\}$, respectively. This study assumes a quasi-static channel setup, where the coherence interval of each channel spans $\tau_c$ symbols. 
Within each coherence interval, the system operation is delineated into four distinct phases: 1) an UL training phase lasting $\tau_p$ symbols, 2) a DL EH phase spanning $\tau_h$ symbols, 3) a DL information transmission phase of $\tau_d$ symbols duration, and 4) an UL information transmission phase extending for $\tau_u$ symbols. These phases collectively facilitate the simultaneous provision of information and power services to the UEs by the APs, within the defined temporal structure of the channel coherence interval. However for our Markov process analysis, we shall focus on the DL EH phase along with UL power requirements in energy differential considerations.

The channel between the $l$-th AP and the $k$-th UE is represented by \(\textbf{g}_{kl} \in \mathbb{C}^{N\times 1}\). Ricean fading channels are assumed, with the channel characteristics remaining the same during each time-frequency coherence interval. Hence, $\gkl$ is expressed as
\begin{equation}~\label{eq:gkl}
    \gkl= \sqrt{\dfrac{\zekl}{\kkl+1}}
    \left(\sqrt{\kkl}\bhkl +\tgkl\right),
\end{equation}
where $\zekl$ represents the large-scale fading coefficient, while $K_{kl}$ is the Ricean factor. Here, the line-of-sight (LoS) component is modeled as $\bhkl = [1,e^{j\pi\sin(\phi_{kl})},\ldots,e^{j(N-1)\pi\sin(\phi_{kl})} ]^T$, while the non-line-of-sight (NLoS) component follows a complex Gaussian distribution as $\tgkl\sim\mathcal{N}_C(\boldsymbol{0}, \qI_N)$. By defining $\beta_{kl}\! \triangleq\!{\dfrac{\zekl}{K_{kl}+1}}$, $\bklq\!\triangleq\! \bkl \kkl$, and
$\bgkl\triangleq \sqrt{\bklq} \bhkl$, we can further simplify \eqref{eq:gkl} as
\begin{equation}
    \gkl= \bgkl +\sqrt{\beta_{kl}}\tgkl.
\end{equation}

Moreover, we assume that the UEs are either static or move slowly, whereas the surrounding environment is dynamic \cite{Hien:Asilomar:2018,Chongzheng}. Therefore, it is justifiable to consider that the large-scale channel characteristics $\bgkl$ change very slowly in order of hundreds or thousands of coherence intervals and are known a priori.   One simple method to estimate $\bgkl$ is by using many channel realizations over time and frequency to compute the average value of the channel. Similarly, $\beta_{kl}$ is also assumed to be known a priori~\cite{Hien:Asilomar:2018,Demir,Wang:JIOT:2020}.

At the beginning of the coherence interval, which lasts for $\tau_p$ symbols, the UL training phase is used for channel estimation. Throughout this duration, all UEs concurrently transmit pre-defined pilot sequences to the APs. The pilot sequence transmitted by $k$-th UE is denoted by $\pmb{\varphi}_{k} \in \mathbb{C}^{\tau_p }$ within the set of available pilot sequences, such that $\lVert\pmb{\varphi}_{k}\rVert^2= \tau_p$.   Ideally, each UE is allocated a pilot sequence that is orthogonal to those of all other UEs, enabling accurate and interference-free channel estimation. However, in practical scenarios with short coherence intervals and/or a large number of UEs, the limited number of orthogonal sequences necessitates the reuse of pilots~\cite{Mohammadi:PROC.2024}. This leads to non-orthogonal pilot transmission, resulting in interference among pilot signals received at each AP—a phenomenon known as pilot contamination, which degrades the quality of channel estimation~\cite{hien}. The case of practical interest is a large
network with $K>\tau_p$ so that more than one UE is assigned to each pilot. Asume  that $\mathcal{P}_k\subset\{1,\ldots,K\}$ denotes the subset of UEs sharing the same pilot signals as the $k$-th UE. The UL training signal received at the $l$-th AP denoted as $\qY_{p,l} \in \mathbb{C}^{N \times \tau_p}$ can be expressed as
\begin{equation}
    \qY_{p,l}=\sum\nolimits_{k\in\KE} \sqrt{p_p}  \textbf{g}_{kl} \pmb{\varphi}^T_{k}+\qN_l,
\end{equation}
where $p_p$ is the pilot transmit power and \(\textbf{N}_l\in \mathbb{C}^{N \times \tau_p}\) is the additive
white Gaussian noise (AWGN) matrix with independent, identically distributed \(\mathcal{N}_C(0,\sigma^2)\) entries. The projection of $\qY_{p,l}$ onto  $\pmb{\varphi}_k$ provides the sufficient statistic for the channel estimate of $\gkl$, given by
\begin{equation}
  \zkl=\dfrac{\qY_{p,l}\pmb{\varphi}^*_{k}}{\sqrt{\tau_p}} = \sqrt{\tauP}  \sui \gil+\qn_{kl},
\label{eq:equalized}
\end{equation}
where $\qn_{kl}\triangleq \dfrac{\qN_l\pmb{\varphi}^*_{k}}{\sqrt{\tau_p}}\sim\mathcal{N}_C(\boldsymbol{0},\sigma^2\qI_N)$. Given $\zkl$ and perfect knowledge of the LoS part, $\bgkl$, the minimum-mean-square-error (MMSE) channel estimate of $\gkl$ is $\hgkl=\bgkl+\thgkl$ ~\cite[Eq. (4)]{Hien:Asilomar:2018}, \cite[Eq. (12.6)]{Kay}, where
\begin{align}
   \thgkl&\triangleq \qC_{\gkl,\zkl}\qC^{-1}_{\zkl} \left(\zkl-\ME\{\zkl\}\right)\nonumber\\
    &= \ckl \left(\sqrt{\tauP} \sui \sqrt{\bil} \tgil +\textbf{n}_{kl}\right),
\end{align}
while $\ckl \triangleq\frac{\sqrt{\tauP} \bkl}{ {\tauP}\sui\bil +\sigma^2}$. The statistical parameters of $\hgkl$ are given by
\vspace{-0.2em}
\begin{subequations}
   \begin{align}
    \ME\{\hgkl\}&=\bgkl,\\
    \MV\{ [\hgkl]_{t}\}&\triangleq \gamkl= \sqrt{\tauP} \bkl \ckl.
\end{align} 
\end{subequations}

It can be observed that the NLoS components of the MMSE estimates for UEs within subset $\mathcal{P}_k$ exhibit a linear correlation, expressed as \(\thgil= \aik \thgkl\) where \(\aik\triangleq \bilsq/\bklsq\). Furthermore, according to the MMSE estimation property, the estimation error, given as $\tekl=  \sqrt{\bkl} \tgkl- \thgkl$, with
$\tekl\sim\mathcal{N}_C(\boldsymbol{0},(\bkl-\gamkl)\qI_N)$, and the the channel estimate $\thgkl$ are uncorrelated \cite{Kay}. Since they are Gaussian distributed, they are also independent.

\section{Downlink Energy Harvesting}\label{sec:DL_EH}

\subsection{Downlink Energy Transmission}
During the DL EH phase, each AP forms precoding vectors based on its
channel estimates to efficiently transmit energy to the UEs. We denote $\wkl \in \mathbb{C}^{N\times 1}$ as the DL precoding vector pertinent to the EH phase for UE $k$ at AP $l$. It is well known that maximum ratio transmission (MRT) serves as an efficacious beamforming strategy for power transfer, particularly when the number of antennas $N$ is significantly large~\cite{almradi2016performance}. The MRT precoding vector at AP $l$ for UE $k$ is given by $\qw_{kl}= \kmrt\hgkl$, where $\kmrt\!=\!\frac{1}{\sqrt{ \ME \left\{\lVert \hgkl \rVert^2 \right\}}}\!=\!\frac{1}{\sqrt{N(\bklq + \gamkl)}}$. The signal transmitted by the $l$-th AP can be expressed as
\begin{align}~\label{eq:xle}
    \xEl&=\suk\sqrt{\Omega_{kl}}\qw^*_{kl}{e}_{k}, 
\end{align}
where $e_k$ denotes the zero-mean, unit-variance energy signal corresponding to UE $k$. It is naturally assumed that the energy symbols for different UEs are independent~\cite{Demir,mohammadi2}. In~\eqref{eq:xle}, $\Omega_{kl}$ represents the power control coefficient, which is chosen in a way to satisfy the total network power constraint, such that 
\begin{equation}
    \sul \mathbb{E}\left\{\lVert \textbf{x}^E_l \rVert^2 \right\}=\sul \suk \Omega_{kl}\leq p_{t,N},
\end{equation}
where $p_{t,N}$ is the total available network power for DL EH phase. During this phase, the signal received at UE $k$ is the combination of coherent signal, multi-user interference, and noise signal. Mathematically speaking, it is represented as
\begin{align}~\label{eq:rE}
   \! \rEk
    =&\sul \suik \!\kmrti \sqrt{\Omega_{il}} \gklt \hgil^*{e}_{i}+n^E_k,\nonumber\\
    =& \underbrace{\sul\! \kappa_{kl} \sqrt{\Omega_{kl}} \gklt \hgkl^*{e}_{k}}_{\text{Coherent Signal}}\nonumber\\
    &+\underbrace{\sul \sum\nolimits_{\underset{i \neq k}{i\in\KE}} \!\kmrti\! \sqrt{\Omega_{il}} \gklt \hgil^*{e}_{i}}_{\text{MUI}}\!+\!\underbrace{n^E_k}_{\text{Noise}},\!
\end{align}
where $n^E_k \sim \mathcal{N}_{\mathbb{C}}(0,\sigma^2)$ is the AWGN at the $k$-th UE. Given that the noise floor is significantly below the threshold useful for EH, the influence of $n^E_k$ during EH phase is considered negligible, in accordance with previous studies~\cite{femenias,Demir,Wang:JIOT:2020,Boshkovska:CLET:2015}. Consequently, the received RF energy at the EH circuit of the $k$-th UE, represented by $\IEk$, can be expressed as
\begin{align}
\label{eq:harvested_power}
   \IEk
    &=\ME \bigg\{\Big\lvert \sul \suik \kmrti\sqrt{\Omega_{il}}  \gklt \hgil^*  {e}_{i} \Big\rvert^2 \bigg\}\nonumber\\
    &=\sul \sum\nolimits_{l'\in \mathcal{L}}\! \suik \kmrti \kmrtip
    \sqrt{\Omega_{il} \Omega_{il'}} 
     \gklt \hgil^* \hgiltp \gklcp,\!
\end{align}
where the expectation is taken over the independent zero-mean energy signals $e_i$. 

\subsection{Non-linear Energy Harvesting}
In order to quantify the harvested energy at UE $k$, a practical non-linear EH model is employed, as delineated in \cite{Boshkovska:CLET:2015}. Thus, the expression for the energy harvested during a single EH interval $\tau_h$ at the $k$-th UE is given by
\begin{equation}
    \EEk= \tau_h \psi_k\left(\Lambda\left[\IEk\right] - \varphi_k\right),
\label{eq:harvested_energy}
\end{equation}
where $\psi_k= I^E_{k,max}/(1-\varphi_k)$, while $I^E_{k,max}$ denotes the maximum output DC power at UE $k$, $\varphi_k=1/{\left(1+e^{a_k b_k}\right)}$ is a constant which guarantees zero-input/zero-output response with $a_k$ and $ b_k$ parameters relevant to EH circuit specifications such as resistance, capacitance and diode turn-on voltage, respectively. Moreover, $\Lambda\left[\IEk\right]$ is the tradition logistic function with respect to the received power $\IEk$, given as
\begin{equation}\label{eq:logistic}
    \Lambda\left[\IEk\right]=\dfrac{1}{1+e^{-a_k\left(\IEk-b_k\right)}}.
\end{equation}

\subsection{Statistics and Probability Distribution of Harvested Energy}
We analyze the stochastic parameters of harvested energy,  which will be used in Section \ref{sec:Markov} to formulate the foundations of probabilistic construction of Markov process based state transitions of UE energy levels. First, the statistics of harvested power $\IEk$, given in \eqref{eq:harvested_power}, are presented, which will lead to the characteristics of the harvested energy $\EEk$ in \eqref{eq:harvested_energy}.

\begin{proposition}~\label{Lemma:EIK}
The statistics of the received RF energy, $\IEk$, can be obtained as
\begin{align}
    \ME\big\{\IEk\big\}
    &\!=\! \sul \slp \suik \!\kmrti\kmrtip\!\sqrt{\Omega_{il}\Omega_{il'}}
    \Xi_{ik,ll'},\!
    \label{eq:mean:final}
    \\
      \MV\big\{\IEk\big\}
    &\!=\! \sul \slp \suik \left(\kmrti \kmrtip\right)^2 {\Omega_{il}\Omega_{il'}}\nonumber\\
    &\hspace{0.5em}\times\left(\Upsilon_{ik,l}^{\mathsf{coh}}\delta_{l,l'} + \Upsilon_{ik,ll'}^{\mathsf{noncoh}}(1-\delta_{l,l'} )\right),
    \label{eq:var:final}
\end{align}
where
\vspace{-0.2em}
\begin{align}
   &\Xi_{ik,ll'}= \delta_{l,l'}N\big(N \bkilsq \oiklsq + \bklq \gamil +\bkl(\bilq+\gamil)\\
    &+ \gamkl (\aiksq (N+1) \gamkl + 2 N  \aik \bkil \oikl -\gamil)\big) \nonumber\\
        &+(1-\delta_{l,l'} )N^2(\bkil \oikl+ \aik \gamkl)(\bkilp \oiklp+\aik \gamklp),
\end{align}
with $\oikl = \bhkl^T \bhilc/N$, $\bkil= \sqrt{\bklq \bilq}$, while $\delta_{l,l'}$ is $1$ when $l=l'$ and $0$ otherwise. Moreover, $ \Upsilon_{ik,l}^{\mathsf{coh}}$ and $\Upsilon_{ik,ll'}^{\mathsf{noncoh}}$ are given in~\eqref{eq:upsilon} at the top of the next page.
\begin{figure*}
\begin{subequations}~\label{eq:upsilon}
  \begin{align}
      \Upsilon_{ik,l}^{\mathsf{coh}} 
                         &=2 N^2 \bkilsq 
                         \big[
                         N \aiksq \bklq \oiklsq \gamkl+N \bkl \oiklsq \bilq+ \aiksq \oiklsq \gamkl ( \bkl+N\gamkl)\!+\!\aiksq \bkl  \gamkl 
                         \big]\! +\! N^2\big[\aiktt \bklqsq \gamklsq 
               \nonumber\\
              & \hspace{-0.65em}
               +\bklsq \bilqsq \big] + 2 \aiksq \gamkl N(N+1) \big[\aiksq  \bklq   \gamkl  \big((N+1)\gamkl + \beta_{kl} \big)+ \bilq \big(  
   (N-1) \bkl\gamkl + \bklsq + 2 \gamklsq\big)\big]
                          \nonumber\\
              & \hspace{-0.65em}
                  +\! N \aiktt \gamklsq \big[ (N\!+\! 1)(N\!+\!2)\gamkl \big((N\!+\!3) \gamkl \! +\!4 (\bkl\!-\!\gamkl) \big)\! +\! (\bkl\!-\!\gamkl)^2(2N\!+\!1)\!-\!(\bkl\!+\!N\gamkl)^2 \big]\!, 
                  \\
     \Upsilon_{ik,ll'}^{\mathsf{noncoh}} &=
     N^2   \oiklsqp\bkilpsq \big[\aiksq\gamkl
     (\bkl + N(\gamkl+\bklq) )
      +
     N\bkl \bilq \big]+
    N^2 \oiklsq \bkilsq\big[ \aiksq \gamklp (\bklp+N (\gamklp+\bklqp))
    \nonumber\\
    &\hspace{-0.65em}+
     N\bklp \bilqp \big]\!+\!N^2\aiksq \big[\bkl  \bilq  \gamklp (\bklp+\bklqp+N\gamklp)
    + \bklp \bilqp \gamkl(\beta_{kl}+\bklq+N\gamkl)\big]
    +
    N^2\bkl \bklp \bilq \bilqp
    \nonumber\\
    &\hspace{-0.65em}+N^2\aiktt  \gamkl \gamklp \big[(\bklq + \bkl)(\bklqp+\bklp)
    +N(\gamkl(\bklp+ \bklqp)+\gamklp(\bkl+\bklq)) 
        \big].
\end{align}
\end{subequations}
\hrulefill
\vspace{-1.5em}
\end{figure*}

\end{proposition}
\begin{proof}
    See Appendix~\ref{Proof:Lemma:EIK}.
\end{proof}
Now, we proceed to calculate the statistics of the harvested energy $\EEk$. By invoking \eqref{eq:harvested_energy}, we first calculate its expectation, which is given by 
\begin{equation}
    \ME\left\{\EEk\right\}= \tau_h \psi_k \left( \ME\left\{\Lambda\left[\IEk\right]\right\}-\varphi_k \right).
\label{eq:exp_eek}
\end{equation}
It can be observed that the computation of $\ME\big\{\Lambda\left[\IEk\right]\big\}$ presents significant complexity. Thus, we use the following approximation\footnote{Considering 
$0 \leq\IEk <b_k$, the logistic function described in~\eqref{eq:logistic} is convex. Since the received input energy typically falls within this interval, Jensen's inequality verifies that this approximation serves as an upper bound for $\ME\big\{\Lambda\left[\IEk\right]\big\}$  with an approximation error of only $0.0081 \%$ as observed in our simulation results.}
\begin{align}
\label{eq:jansen_inequality}
         \ME\big\{\Lambda\left[\IEk\right]\big\}\approx \Lambda\left[\ME\left\{\IEk\right\}\right]=\dfrac{1}{1+e^{-a_k\left(\ME\left\{\IEk\right\}-b_k\right)}}.
\end{align}
Therefore, $\ME\left\{\EEk\right\}$ can be expressed as 
\vspace{0.2em}
\begin{align}
\label{eq:jansen_inequality2}
        &\ME\left\{\EEk\right\}
         =\psi_k\left(\Lambda\left[\ME\left\{\IEk\right\}\right]-\varphi_k\right)+\acute{\epsilon}_k,
\end{align}
where \(\acute{\epsilon}_k\) is the approximation gap. Note that the Taylor series expansion of \eqref{eq:exp_eek} at $\ME\left\{\IEk\right\}$ also yields the same expression as in \eqref{eq:jansen_inequality2}, since the higher-order terms become asymptotically negligible.

Now, the variance of the harvested energy at the output of the non-linear EH circuit, i.e.,~\eqref{eq:harvested_energy}, can be obtained as
\begin{equation}
        \MV\left\{\EEk\right\}\!=\!(\tau_h \psi_{k})^2 \left(\ME\left\{\big(\Lambda\left[\IEk\right]\big)^2\right\}\!-\!\left(\ME\left\{\Lambda\left[\IEk\right]\right\}\right)^2\right).
\label{eq:var_eek}
\end{equation}

Using a Taylor series expansion of $\Lambda\left[\IEk\right]$ around the expectation $\ME\left\{\IEk\right\}$, we can approximate $\ME\Big\{\big(\Lambda\left[\IEk\right]\big)^2\Big\}$ in \eqref{eq:var_eek} as
\vspace{0.5em}
\begin{align}
        &\ME\Big\{\big(\Lambda\left[\IEk\right]\big)^2\Big\}
        \!=
        \! \ME\Big\{\!\Big(\Lambda\!\left[\ME\left\{\IEk\right\}\right] +\frac{\partial \Lambda\left[\ME\left\{\IEk\right\}\right]}{\partial \IEk }\nonumber\\
        &\hspace{0em}
        \times\left(\IEk\!-\ME\left\{\IEk\right\}\right)
        +\dfrac{1}{2}\frac{\partial^2 \Lambda\left[\ME\left\{\IEk\right\}\right]}{\partial (\IEk)^2 }\left(\IEk-\ME\left\{\IEk\right\}\right)^2\Big)^2 \Big\},\nonumber\\     
        &\hspace{6em} \approx\! \MV\left\{\IEk\right\}\!\left(\bigg(\frac{\partial \Lambda\left[\ME\left\{\IEk\right\}\right]}{\partial \IEk }\bigg)^2
        +\!\Lambda\!\left[\ME\left\{\IEk\right\}\right]
        \right.
        \nonumber\\
        &\hspace{6em}
        \left.\times
        \frac{\partial^2 \Lambda\left[\ME\left\{\IEk\right\}\right]}{\partial (\IEk)^2 }\right)\!+ \!\left(\Lambda\left[\ME\left\{\IEk\right\}\right]\right)^2.
\label{eq:firstVarexp}
\end{align}
Substituting \eqref{eq:firstVarexp} into \eqref{eq:var_eek}, the variance of the harvested energy is given by
\vspace{1em}
\begin{align}
    \!\MV\left\{\EEk\right\}\!=& 
    (\tau_h \psi_{k})^2 \MV\left\{\IEk\right\}
    \bigg( \frac{\partial^2 \Lambda\left[\ME\left\{\IEk\right\}\right]}{\partial (\IEk)^2 }\left[\ME\left\{\IEk\right\}\right]\nonumber\\
    &+\bigg(\frac{\partial \Lambda\left[\ME\left\{\IEk\right\}\right]}{\partial \IEk }\bigg)^2\bigg)+\acute{\chi}_k,
\label{eq:var_eek_final}
\end{align}
where \(\acute{\chi}_k\) is the approximation error, and
\begin{subequations}
   \begin{align}
        \!\!\frac{\partial \Lambda\!\left[\ME\left\{\IEk\right\}\right]}{\partial \IEk } &= \dfrac{a_k e^{-a_k\left(\ME\left\{\IEk\right\}-b_k\right)}}{\Big(1+e^{-a_k(\ME\left\{\IEk\right\}-b_k)}\Big)^2}\nonumber\\
        &= a_k \Lambda\!\left[\ME\left\{\IEk\right\}\right]\left(1-\Lambda\!\left[\ME\left\{\IEk\right\}\right]\right),
        \\
        \!\!\frac{\partial^2 \Lambda\!\left[\ME\left\{\IEk\right\}\right]}{\partial (\IEk)^2 } &= \dfrac{a^2_k e^{-a_k\left(\ME\left\{\IEk\right\}-b_k\right)}\!\Big(1\!-\!e^{-a_k(\ME\left\{\IEk\right\}-b_k)}\Big)}{\Big(1+e^{-a_k(\ME\left\{\IEk\right\}-b_k)}\Big)^3},\nonumber\\
        &\hspace{-4em}=a^2_k \Lambda\!\left[\ME\left\{\IEk\right\}\right]\left(1-\Lambda\!\left[\ME\left\{\IEk\right\}\right]\right)\left(1-2\Lambda\!\left[\ME\left\{\IEk\right\}\right]\right).
\end{align}
\end{subequations}
We now concentrate on the probability distribution function (PDF) of the $\EEk$, utilizing the previously derived expressions for the mean and variance as indicated in \eqref{eq:jansen_inequality} and \eqref{eq:var_eek_final}. Typically, the energy storage from a potential source follows an exponential distribution, as described in ~\cite[Eq. (27.33)]{sadiku}. In the current scenario, multiple signals act as potential energy sources, as depicted in \eqref{eq:rE}.  Consequently, the composite harvested energy $\IEk$ at the input of the EH circuit  will approximately follow a Gamma distribution, representing an effective sum of independent exponential distributions \cite{kusal,tavana}. By observing \eqref{eq:exp_eek}, the PDF of $\EEk$ is dependent upon the logistic function $\Lambda\left[\IEk\right]$ while all other parameters are constants. The second order Taylor series expansion of $\Lambda\left[\IEk\right]$ around the expectation $\ME\left\{\IEk\right\}$ present in calculating \eqref{eq:firstVarexp}, can be written as
\begin{align}\label{eq:taylor_approx}
        \Lambda\left(\IEk\right)\!\approx&\Lambda\left[\ME\left\{\IEk\right\}\right]+a_k \Lambda\!\left[\ME\left\{\IEk\right\}\right]\left(1-\Lambda\!\left[\ME\left\{\IEk\right\}\right]\right)\nonumber\\
        &\hspace{-0.2em}\times\!\left(\IEk\!-\!\ME\left\{\IEk\right\}\right)\!+\! \dfrac{a^2_k}{2} \Lambda\!\left[\ME\left\{\IEk\right\}\right]\left(1\!-\!\Lambda\!\left[\ME\left\{\IEk\right\}\right]\right)\nonumber\\
        &\hspace{-0.2em}\times\!\left(1-2\Lambda\!\left[\ME\left\{\IEk\right\}\right]\right)\left(\IEk\!-\ME\left\{\IEk\right\}\right)^2.
\end{align}
Now, we take the variance of this function as
\begin{equation}\label{eq:var_lambda}
    \!\MV\{\Lambda\!\left(\IEk\right)\}\! \approx\! \Big(\! a_k \Lambda\!\left[\ME\left\{\IEk\right\}\right]\!\left(1\!-\!\Lambda\!\left[\ME\left\{\IEk\right\}\right]\right)\!\!\Big)^{2}\! \MV\left\{\IEk\right\}.\!
\end{equation}
Using the moments evaluated in \eqref{eq:jansen_inequality} and \eqref{eq:var_lambda}, we estimate the shape and scale of the Gamma distribution representation of $\Lambda\left(\IEk\right) \sim \Gamma\left(k_{\Lambda}, \theta_{\Lambda}\right)$, as follows
\begin{align}\label{eq:gam_stat}
    \begin{split}
        k_{\Lambda}=\dfrac{\left(\ME\left\{\Lambda\left(\IEk\right)\right\}\right)^2}{\MV\left\{\Lambda\left(\IEk\right)\right\}},\quad
        \theta_{\Lambda}=\dfrac{\MV\left\{\Lambda\left(\IEk\right)\right\}}{\ME\left\{\Lambda\left(\IEk\right)\right\}}.
    \end{split}
\end{align}
Now, we extend this analysis from $\Lambda\left(\IEk\right)$ to $\EEk$. Using \eqref{eq:harvested_energy}, we express the relationship of PDFs of these random variables as,
\begin{equation}
    f\left(\EEk\right) = \dfrac{1}{\tau_h \psi_k}f_{\Lambda}\left(\dfrac{\EEk}{\tau_h \psi_k}+\varphi_k\right).
\end{equation}
This equation shows that $\EEk$ is Gamma distributed $\sim \Gamma\left(k_{\Lambda},\tau_h \psi_k\theta_{\Lambda}\right)$ with support shifted by $\tau_h \psi_k\varphi_k$.
For the $k$-th UE, we can approximate the shape and scale parameters of the Gamma distribution using the mean and variance of $\EEk$ calculated in \eqref{eq:jansen_inequality} and \eqref{eq:var_eek_final} as
\begin{align}\label{eq:gam_stat}
    \begin{split}
        k_{k}=\dfrac{\left(\ME\left\{\EEk\right\}\right)^2}{\MV\left\{\EEk\right\}},\quad
        \theta_{k}=\dfrac{\MV\left\{\EEk\right\}}{\ME\left\{\EEk\right\}}.
    \end{split}
\end{align}

Based on these statistical parameters, we can rewrite the PDF and the cumulative distribution function (CDF) of $\EEk$ as
\begin{subequations}
 \begin{align}
    f\left(\EEk;k_k,\theta_k\right) &\approx \dfrac{1}{\Gamma(k_k)\theta^k} \EEk^{(k_k-1)}e^{-\EEk/\theta_k},\label{eq:gam_pdf}\\
    F\left(\EEk;k_k,\theta_k\right) &\approx \dfrac{1}{\Gamma(k_k)} \gamma \left(k_k, \dfrac{\EEk}{\theta_k} \right).\label{eq:gam_cdf}
\end{align}   
\end{subequations}


\section{Markov Process Based Energy State Transitions} \label{sec:Markov}

In this section, we develop a Markov process-based model for probabilistic transitions among multiple discrete states of UE energy storage. A Markov process is aptly suited for characterizing stochastic processes where the future state is contingent solely upon the current system state. Within the framework of the proposed WPT-based CF-mMIMO system, the EH dynamics can be effectively represented using a MC as the UE energy storage state at the subsequent coherence interval is exclusively dependent on its present energy state and energy differential between the energy harvested by the APs and the energy consumed during UL training and data transmission phases. Figure~\ref{fig:Markov_chain} provides a schematic depiction of this process, with the nodes symbolizing the discrete UE energy states and the directional arrows illustrating the transition probabilities between the states.

\begin{figure}[t]
    \centering \includegraphics[width=0.9\columnwidth]{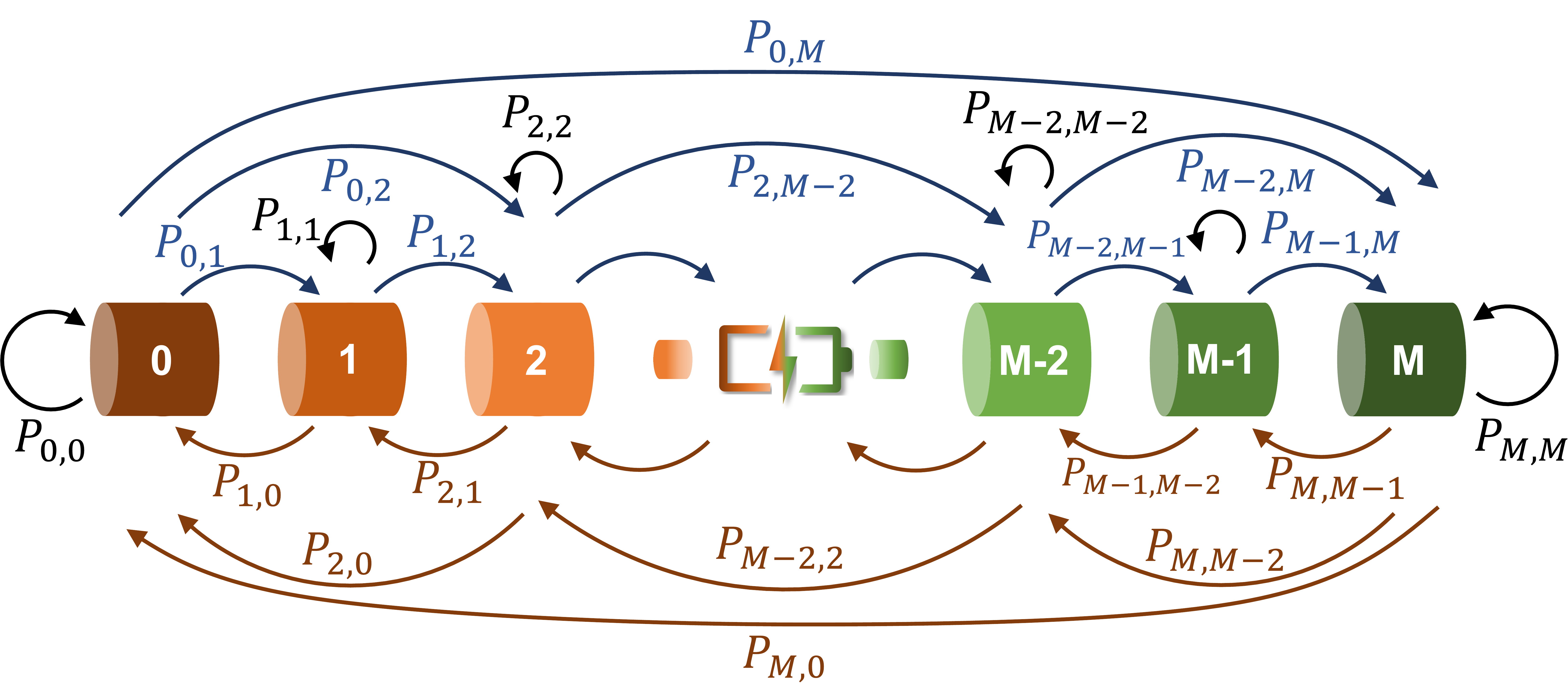}
     \vspace{0em}
    \caption{\small MC-based energy state transitions.\normalsize}
    \vspace{0.4em}
    \label{fig:Markov_chain}
\end{figure}

\subsection{Dynamics of Energy State Transitions}
Each UE is postulated to have a finite energy storage of total capacity $E_f$. This continuous energy storage is segmented into $M$ discrete levels, each representing an energy quantum of $\delta E = E_f/M$. Since the UE energy storage is fundamentally continuous in nature, we try to develop better understanding about its temporal evolution by  quantizing it into an arbitrarily large number of discrete states, such that the discrete characterization approaches the continuous one \cite{kusal}. Given the memoryless nature of Markov processes, which relies on the postulate that the future state is independent of the past transition history and is solely dependent on the present state, the probability of the initial state at the inception of observation process follows a uniform distribution. Thus, the initial probability $\pi_0(m) \triangleq\Pr(X_0=m) \,\, \forall \,\, E^0_k \in \{E_1,E_2, \hdots, E_f\}$ of any particular discrete state $m$ at time $t=0$ is $1/M$. Furthermore, the transition probabilities from any given state $i$ can be represented as $\{P_{i,0},P_{i,1}, \hdots, P_{i,i-1}, P_{i,i}, P_{i,i+1}, \hdots, P_{i,M}\}$. 

Considering a scenario where the UE energy level is in state $i$ at time $t$, the subsequent potential states are influenced by the energy differential $\dEk$ within the storage during a single coherence intervals $\tau_c$ as shown in Fig. \ref{fig:State_transition_probabilities}. This differential $\dEk$ is derived from the net balance of energy harvested ($\EEk$) during the DL EH phase, and the energy expended during both the UL pilot training and information transfer, quantified as $\ECk \triangleq \tau_p p_p + \tau_u p_u$. Consequently, $\dEk$ is defined as $\dEk= \EEk - \ECk = \EEk - (\tau_p p_p + \tau_u p_u)$. Moreover, at the $m$-th interval of this Markov process, the total energy of UE battery comprises the sum of the initial battery  $E^0_k$ and successive energy differentials $\Delta E^i_k$, given as, $E^n_k= E^0_k + \sum^n_{i=1} \Delta E^i_k$.  For analytical tractability, we assume that the expected value of energy differential is significantly smaller than the discrete energy slot increment as $\lvert\ME\{\Delta E_k\}\rvert\!\! \ll\!\! E_f/M $, which will be validated in Section \ref{sec:results}. After this reasonable assumption, we can reduce the feasible set of state transitions to three states: $\{P_{i,i-1}, P_{i,i}, P_{i,i+1}\}$.


\begin{figure}[t]
    \centering
    \includegraphics[width=0.9\columnwidth]{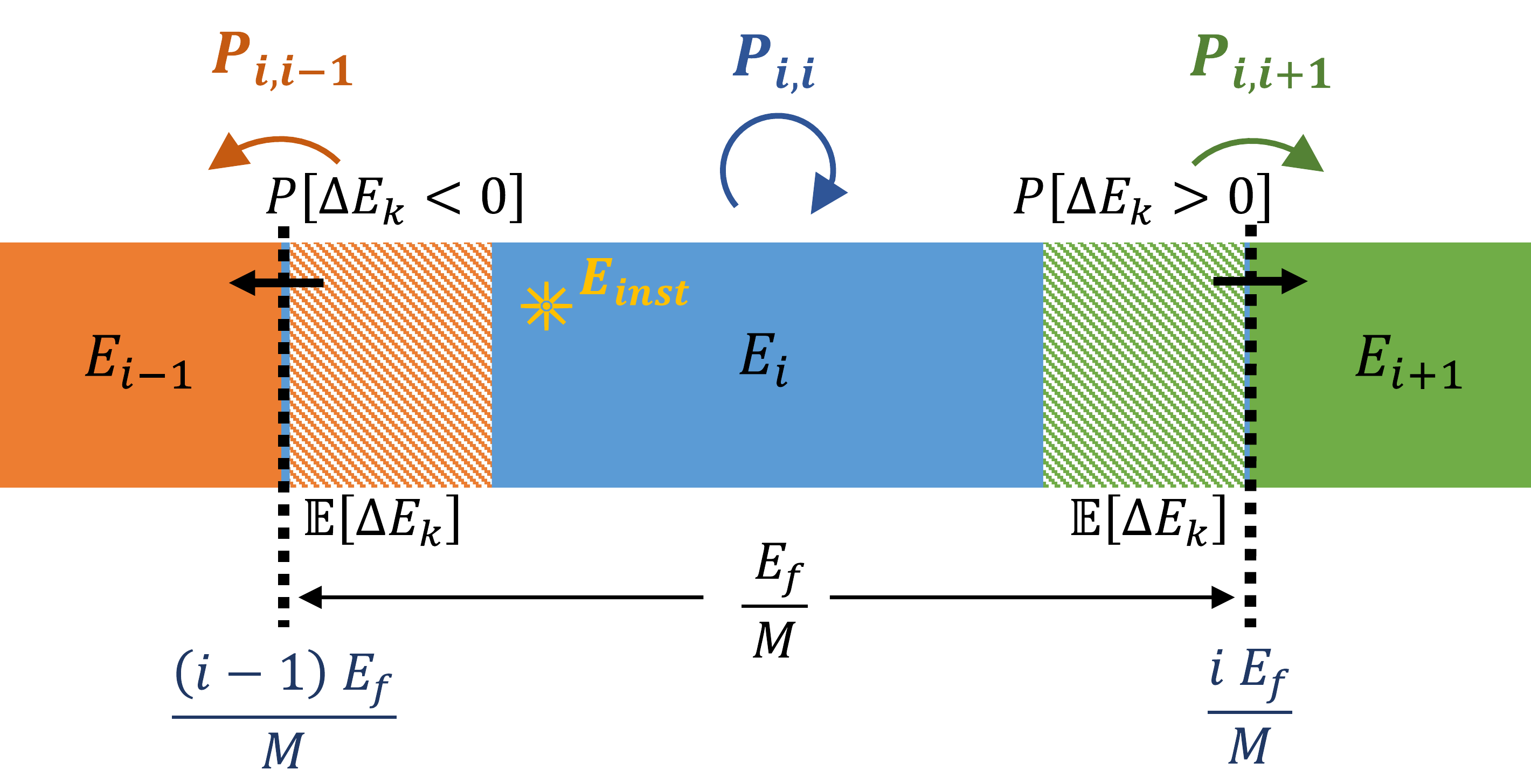}
     \vspace{-0.7em}
    \caption{\small Transition probability cases from a particular energy state.\normalsize}
     \vspace{0.5em}
    \label{fig:State_transition_probabilities}
\end{figure}

In Fig. $\ref{fig:State_transition_probabilities}$, the parameter $P_{i,i+1}$ denotes the transition probability that a particular UE will have accumulated sufficient energy during $\tau_c$ resulting in positive energy differential $\dEk$ to jump from energy state $i$ to the adjacent higher state $i+1$, despite the energy consumption $\ECk$. In a similar manner, we can explain the negative transition probability  $P_{i,i-1}$, that the harvested energy $\EEk$ is less than the consumed energy $\ECk$, such that the UE energy level transits to a lower state $i-1$. However, for the largest portion of time, the UE energy level is self-transitioning in the same state $i$, shown by $P_{i,i}$ , which means that the net energy difference $\dEk$ is not significant enough to influence the instantaneous energy level to shift to adjacent states. 

It is imperative to observe by visual inference of Fig. \ref{fig:State_transition_probabilities} that the instantaneous energy level ($E_{inst}$, yellow star) in state $i$ during a particular coherence interval is a fundamental measure in the determination of the next transition $j \in \{i-1,i+1\}$. It can be deduced that when the instantaneous level is close to the energy state boundaries within $\ME\{\dEk\}$ region (shaded green and orange regions), the likelihood of transitioning to adjacent states increases. Further quantifying this phenomenon, the conditional probability of state transition can be determined by the two statistical parameters of $\dEk$: the mean $\ME\{\dEk\}$ and the probability of positive or negative energy differential, $\Pr(\dEk\!\leq\!0)$ or $\Pr(\dEk\!>\!0)$. The magnitude of the mean energy differential $\ME\{\dEk\}$ determines the span of the region in vicinity of state $i$ boundaries (shaded green and orange regions), where it is highly likely that the instantaneous energy level ($E_{inst}$) during interval $m-1$, combined with the expected energy differential in next coherence interval $m$, will jump to neighboring state $j$ in interval $m$. Given this condition, if $E_{inst}$ is in the green shaded region in interval $m-1$ and positive energy differential $\Pr(\dEk\!\geq\!0)$ occurs in interval $m$, the energy state of UE storage will transit to the next higher state $E_{i+1}$ (solid green). On the other hand, if $E_{inst}$ is in the orange shaded region and negative energy differential $\Pr(\dEk\!\leq\!0)$ happens, the energy state will jump to lower state $E_{i,i-1}$. For the case where $E_{inst}$ in state $i$ is outside these regions (solid blue region) regardless the probability of energy differential, self-transition $P_{i,i}$ is expected to happen in the next coherence interval.  

\subsection{Stochastic Construction of Energy Transitions}
Now, we develop a quantitative representation for the probabilistic state transition of the above explained discrete energy state model for UE energy storage. We first discuss the self-transition case $P_{i,i}$ within the state $E_i$ in Fig. \ref{fig:State_transition_probabilities}. If the instantaneous energy level $E_{inst}$ within state $E_i$ in interval $m-1$ lies in between $\big((i-1)E_f/M + \dEk\big)$ and $\big((i+1)E_f/M - \dEk\big)$ (solid blue region), then the UE storage in next coherence interval $m$ with still maintain state $E_i$. The associated probability of this event $P_{i,i}$ is $(1- M \ME\{\Delta E_k\}/{E_f})$. This formulation implies that  $E_{inst}$ does not approach the proximity of state boundaries with $\dEk$ margin on average. On the other hand, the probabilities of state transitions $P_{i,i-1}$ and $P_{i,i+1}$ are contingent upon the complementary event of self-transition $\bar{P}_{i,i}$. Given this conditional event, a negative energy state transition will occur when the energy consumption $\ECk$ is more than the harvested energy $\EEk$ $\left(\EEk \leq \ECk\right)$. Utilizing the CDF expression of the harvested energy in \eqref{eq:gam_cdf}, it can be shown that
\begin{align}
    \begin{split}
        \Pr\left(\Delta E_k \!\leq\! 0 \right)&\!=\! \Pr\left(\EEk \!\leq \!\ECk \right)\!\approx\!\dfrac{1}{\Gamma(k_k)} \gamma \left(k_k, \dfrac{\ECk}{\theta_k} \right).\! 
    \end{split}
\label{eq:neg_prob}
\end{align}

On the other hand, the positive state transition probability can be expressed as
\begin{equation}\label{eq:pos_prob}
    \!\Pr\!\left(\Delta E_k \! >\! 0 \right)\!=\! \Pr\left(\EEk \!>\! \ECk \right)\approx\! 1\!-\!\dfrac{1}{\Gamma(k_k)} \gamma \left(k_k, \dfrac{\ECk}{\theta_k} \right)\!.\!
\end{equation}


To encapsulate the preceding exposition, the conditional probability of transition from a particular state $i$ during the $m\!-\!1$ interval to neighboring state $j$ in the subsequent interval $m$ for UE $k$ can be denoted as
\begin{align}
\centering
    \label{eq:transition_prob}
    P^k_{i,j}&= \Pr\left(X^k_{m}=j \Bigl\lvert X^k_{m-1}=i \right)\nonumber\\
    &=\!  \Pr\!\left(\!\dfrac{(j-1)E_f}{M}\!\le\! E^{m}_k\!\le\! \dfrac{j E_f}{M} \Bigl\lvert \dfrac{(i\!-\!1)E_f}{M}\!\!\le\!\! E^{m-1}_k\!\le\! \dfrac{i E_f}{M}\!\right)\nonumber\\
    &\approx
    \begin{cases}
        1- \dfrac{M \ME\{\Delta E_k\}}{E_f} &  ~ j\!=\!i \\
        \dfrac{M \ME\{\Delta E_k\}}{E_f \Gamma(k_k)}  \gamma \left(k_k, \dfrac{\ECk}{\theta_k} \right) & ~ j\!=\!i+\!1 \\
        \dfrac{M \ME\{\Delta E_k\}}{E_f} \left( 1-\dfrac{1}{\Gamma(k_k)} \gamma \left(k_k, \dfrac{\ECk}{\theta_k} \right)\right) &  ~ j\!=\!i-\!1, 
     \end{cases}
\end{align}
where $\ME\{\Delta E_k\}=\tau_h \psi_k \left( \ME\left\{\Lambda\left[\IEk\right]\right\}-\varphi_k \right)-\ECk$, while $\ME\left\{\Lambda\left[\IEk\right]\right\}$ is given in~\eqref{eq:jansen_inequality}. To develop physical intuition for these probability cases, we should analyze them in the context of two conditional events: 1) the position of $E_{inst}$ relative to state boundaries. 2)
the polarity of the energy differential $\dEk$.
The first event determines the probability of a self-transition. The second event, which is conditioned on the first, determines whether the transition will be positive or negative. This depends on whether the harvested energy exceeds the energy consumption or vice versa.

Utilizing the Markov property of state independence and the definition of conditional transition probabilities, the analysis can be expanded to elucidate a critical performance metric for EH systems: the probabilities of state transitions over $n$ coherence intervals. This metric offers significant insights into the dynamics of energy accumulation over multiple coherence intervals. Assuming an initial state $q$ of UE $k$ with initial probability $\pi^k_0(q)$, this parameter can be represented as
\begin{align}
    \pi^k_n(j)&=\Pr\left(X^k_{n}=j \right)= \Bigg ( \prod^n_{m=1} \sum^M_{i=1} P^k_{i,j} \Bigg ) \pi^k_0(q).
\label{eq:trans_prob}
\end{align}

\section{Max-Min UE Energy by Power Allocation Schemes}
\label{sec:optimization}
In IoT networks utilizing centralized EH, maintaining adequate energy levels across all the connected devices is critical for continuous operational capabilities, while performing tasks such as UL pilot training and data transmission. Particularly, devices located farther from the APs tend to consume their energy reserves more rapidly due to higher power consumption, which necessitates a more focused energy distribution strategy. Therefore, it is imperative that more power should be focused on these depleted energy UEs. However, the power allocation schemes should be dynamic in nature, monitoring the storage state of all UEs on specific intervals, which we define as the PA period represented as $q$. This approach ensures that all devices sustain operation above a certain minimum energy level after numerous coherence intervals of network operation. Considering this above discussion, we formulate a realistic problem statement and then, present two PA schemes to address this issue.
\subsection{Problem Formulation}
The proposed Markov process-based CF-mMIMO network can provide a pragmatic EH solution for the above discussed practical industrial challenge. It can leverage an efficient power coefficient allocation scheme during the EH phase, which preferentially favors the UEs with least energy available in their storage units. By availing of the Markov property of single state dependence, we can devise a problem to maximize the energy differential during the PA period for the UE with minimum energy level among the network UEs at the end of previous PA period. We can now formulate this problem as
\begin{subequations}\label{eq:optimization_problem_1}
     \begin{align}
         &\max_{\left\{\opkl\right\}}\,   \ME\left\{\Delta E_{\bar{k}}^{(q)}\right\} \quad \text{where} \quad \bar{k} = \arg\min_{k \in \KE} E^{(q-1)}_k \\
         & \text{s.t.}\quad 0\leq \opkl, \quad k\in\KE, ~ l\in\LA,\\
         & \quad \quad \sul \suk  \, \opkl \leq p_{t,N},
     \end{align}
\end{subequations}
where $\ME\left\{\Delta E_{\bar{k}}^{(q)}\right\}$ is the expected energy differential during PA period $q$ for a particular UE $\bar{k}$, which has the lowest energy level $E_{\bar{k}}$ among all $K$ UEs at the end of PA period $(q-1)$. 
\subsection{Heuristic PA Scheme}
Heuristic techniques can provide adaptable and computationally efficient power allocation solutions for complex resource management challenges in various communication networks. These schemes enable close-to-optimal power distribution while accounting for practical constraints such as total power availability and dynamic network conditions. By reducing the computational overhead and ensuring real-time decision-making, heuristic approaches can enhance the overall efficiency and practicality of PA process in contemporary wireless networks. Here, we propose CCPA which allocates the DL EH power on the basis of the relative channel characteristics that can influence the harvested energy. In this rule-based method, we focus the total network power $p_{t,N}$ onto the particular UE $\bar{k}$ with minimum energy level, while accounting for its channel estimates $\hat{\qg}_{\bar{k} l}$ from all the APs. The PA for this approach in period $q$ can be mathematically represented as
\begin{align} \label{eq:heuristic2}
    \Omega^{(q)}_{kl,H}=
    \begin{cases}
        & \dfrac{\lVert \hat{\qg}_{{\bar{k}}l}\rVert_2}{\sum_{l \in L}\lVert \hat{\qg}_{{\bar{k}}l}\rVert_2} p_{t,N} \quad \forall \quad k=\bar{k} \,\,,\,\, k\in\LA \\
        & 0 \hspace{3.17cm}\forall \quad k\neq \bar{k}.
    \end{cases}
\end{align}
\vspace{-0.7cm}
\subsection{Optimized PA Scheme}
The motivation of developing  optimization algorithms lies in constructing precise and systematic strategies for resource allocation in communication networks. These techniques ensure efficient utilization of power resources by handling complex optimization problems to derive near-optimal solutions which can serve as the performance benchmarks for other PA schemes. In this section, we propose a comprehensive optimization algorithm to maximize the minimum energy level of the proposed CF-mMIMO network.
\subsubsection{Problem Reformulation} 
For devising an effective optimization algorithm, we shall transform the original objective function to an form which has direct dependence on the optimization variables. Using the definition of $\dEk$ and \eqref{eq:jansen_inequality}, we approximate $\ME\left\{\Delta E_{\bar{k}}^{(q)}\right\}$ as
\begin{equation}
    \ME\left\{\Delta E_{\bar{k}}^{(q)}\right\} \approx \tau_h \psi_{\bar{k}} \left( \dfrac{1}{1+e^{-a_{\bar{k}}\left(\ME\left\{I_{\bar{k}}^{(q)}\right\}-b_{\bar{k}}\right)}}-\varphi_{\bar{k}} \right)-E_{\bar{k}}^C.
\label{eq:exp_dek}
\end{equation}

The logistic expression given in \eqref{eq:exp_dek} increases monotonically with $\ME\left\{I_{\bar{k}}^{(q)}\right\}$ over the interval $0  \! \leq \ME\left\{I_{\bar{k}}^{(q)}\right\}\! < b_{\bar{k}}$. So, it is sufficient to maximize $\ME\left\{I_{\bar{k}}^{(q)}\right\}$ over $\opkl$ to obtain the optimized solution of the original optimization problem \eqref{eq:optimization_problem_1}. We can now re-formulate the optimization problem with substitution of $\ME\left\{I_{\bar{k}}^{(q)}\right\}$ given in \eqref{eq:mean:final} as
\begin{subequations}\label{eq:optimization_problem_2}
     \begin{align}
         &\max_{\left\{\opkl\right\}}\,  \sul\! \sum\nolimits_{l' \in \mathcal{L}}\! \sum\nolimits_{i \in \mathcal{K}} \!\kmrti\kmrtip\sqrt{\opil \opild}
        \Xi^{(q)}_{i{\bar{k}},ll'}  \quad  \\
         & \text{s.t.}\quad 0\leq \opkl, \quad k\in\KE, ~ k\in\LA,\\
         & \quad \quad \sul \suk  \, \opkl \leq p_{t,N},
     \end{align}
\end{subequations}
where $\bar{k} = \arg\min_{k \in \KE} E^{(q-1)}_k$.

\begin{Corollary}
  The optimization problem in \eqref{eq:optimization_problem_2} is concave.  
\end{Corollary}
\begin{proof} 
We first evaluate the Hessian matrix $\textbf{H}^{(q)}_{\okl}$ of the objective function  given in \eqref{eq:optimization_problem_2} whose entries can expressed as
\begin{align}
    \!\dfrac{\partial^2 \ME\!\left\{\!I_{\bar{k}}^{(q)}\!\right\}}{\partial \Omega^{(q)}_{il}\partial \Omega^{(q)}_{il'}}= \sul\! \sum_{l' \in L}\! \sum_{i \in K} \!\kmrti\kmrtip
    \Xi^{(q)}_{i{\bar{k}},ll'} H^{(q)}_{il,il'},  
\end{align}
where
\begin{align}
H^{(q)}_{il,il'}\!\!=\!\! \sqrt{\Omega^{(q)}_{il'}\Omega^{(q)}_{il}}\! \!\left(\!\dfrac{1}{4 \Omega^{(q)}_{il}\Omega^{(q)}_{il'}}\!-\! \dfrac{1}{2 \left(\Omega^{(q)}_{il}\right)^2}\delta_{l,l'}\!\right)\!.\nonumber
\end{align}
To prove the concavity of $\textbf{H}^{(q)}_{\okl}$, we consider the quadratic form $\qv^T \textbf{H}^{(q)}_{\okl} \qv$ for any arbitrary vector $\qv=[v_{11}, v_{12}, \hdots,v_{il},\hdots, v_{KL}]$, which should be less than zero for a negative semi-definite Hessian matrix. That is,
\begin{equation}
    \dfrac{\sqrt{\Omega^{(q)}_{il'}\Omega^{(q)}_{il}}}{4}\!\! \left(\!\! \left(\!\dfrac{v_{il}}{\Omega^{(q)}_{il}}\!+\!\dfrac{{v_{il'}}}{\Omega^{(q)}_{il'}}\!\right)^2\!\!\!-\!\left(\!\!\dfrac{v^2_{il}}{\left(\Omega^{(q)}_{il}\right)^2}\!+\!\dfrac{v^2_{il'}}{\left(\Omega^{(q)}_{il'}\right)^2}\!\!\right)\!\right)\!\leq 0.
\end{equation}
By Cauchy-Schwarz inequality $\langle a,b \rangle^2 \leq \lVert a \rVert ^2\lVert b \rVert^2 $ over the $\mathbb{R}^2$ vector space \cite{bachman}, the quadratic form of Hessian is negative. Hence, the concavity of the optimization function is proved. 
\end{proof}
\subsubsection{Optimization Process}

In this subsection, the optimization solution in \eqref{eq:optimization_problem_2} will be presented based on transformation to second-order cone programming (SOCP) and iterative application of IPM. The original problem can be reformulated by using second-order cone constraints and linearizing the geometric mean of the power coefficients as
\begin{subequations}\label{eq:optimization_problem_3}
     \begin{align}
         &\min_{\left\{\opkl\right\}}\,  - \sul\! \sum\nolimits_{l' \in \LA}\! \sum\nolimits_{i \in \KE} \!\kmrti\kmrtip \GMp
        \Xi^{(q)}_{i{\bar{k}},ll'}  \,\, 
        \\
         & \text{s.t.}\,\, \left\|\! \begin{pmatrix} \!\! 2 \GMp \\ \opil\!\! \!- \!\opild \! \!\end{pmatrix}\! \right\|\!\! -\!\!\bigg(\!\opil\!\! +\! \opild \!\!\bigg)\!\!\leq \! 0, \,\,\forall \left\{\opil\!\!,\opild \!\right\} \!\in\! \op\!\!, \label{eq:optimization_problem_3_1}\\ 
         & \quad \quad \sul \suk  \, \opkl - p_{t,N}\leq 0 , \, \, \forall \, \opkl\!\! \in\! \op,\, l\in\LA, \label{eq:optimization_problem_3_2}\\
         & \quad \quad \opkl + \opklv- p_{t,N} = 0, \quad k\in\KE \label{eq:optimization_problem_3_3}\\
         & \quad \quad \opkl \geq 0 , \quad \opklv \geq 0,
     \end{align}
\end{subequations}
where  $\bar{k} = \arg\min_{k \in \KE} E^{(q-1)}_k$.

Here, the constraint given in \eqref{eq:optimization_problem_3_1} transforms the problem in the second-order conic domain $\left(\lVert \bm{u}\rVert\leq v\right)$, while ensuring condition $\sqrt{\opil\opild}\!\leq\! \GMp$. Moreover, the introduction of the slack variables $\opklv$ ensures that the inequality constraint  \eqref{eq:optimization_problem_3_2} is always satisfied during the optimization process. For the ease of notation, we omit superscript $(q)$ over the optimization variables, where it is understood that these parameters are valid for a particular PA period. Now, we recast the optimization problem in terms of the Lagrangian function for problem \eqref{eq:optimization_problem_3} with a logarithmic barrier-augmented objective function as
\begin{equation}\label{eq:optimization_problem_4}
    \min_{\bx}\,  \mathcal{L}^{(q)}\!\left( \bx ,\bu,\bl\right)\,\, \text{where} \quad \bar{k}  = \arg\min_{k \in \KE} E^{(q-1)}_k,
\end{equation}
where $\bx=\left\{\Omega^{(q)}_{{\bar{k}}l} ,\breve{\Omega}^{(q)}_{{\bar{k}}l} , \GMp\right\}$, $\bu=\left\{\mu^{(q)}_{\upsilon},\mu^{(q)}_{\Omega}\right\}$ with Lagrange multiplier $\bl$, while the Lagrangian function is defined as:
\begin{align}
    \mathcal{L}^{(q)}\!\left(\bx ,\bu, \bl\! \right)\! =  \! \fx \!+\!\!  \bu^T \bgx \!+ \! \! \bl^T h\left(\bx \right). 
\end{align}
 Here, the primary objective function is represented as
 \begin{align*}
    \fx\!=\!- \sul\! \sum\nolimits_{l' \in \mathcal{L}}\! \sum\nolimits_{i \in \mathcal{K}} \!\kmrti\kmrtip  \GMp \Xi^{(q)}_{i\bar{k},ll'}. 
 \end{align*}
 The logarithmic barrier function for the inequality constraints \eqref{eq:optimization_problem_3_1}, \eqref{eq:optimization_problem_3_2} is $\bgx=\bigg[g ^{(q)}_{\upsilon}\left(\bx \right)$, $ g^{(q)} _{\Omega}\left(\bx \right)\bigg]^T$ for $g^{(q)} _{\upsilon}\!\left(\bx\right)\!\!=\! \left(\GMp\right)^2\!-\!\opil\opild$ and $g^{(q)}_{\Omega}\!\left(\bx\right)=p_{t,N}-\sul \suk  \, \opkl $ with barrier parameter $\bu=\left[\mu^{(q)}_{\upsilon},\mu^{(q)}_{\Omega}\right]^T$ for $\mu^{(q)}_{\upsilon}=-1/\left(t*g^{(q)} _{\upsilon}\!\left(\bx\right)\right)$, $\mu^{(q)}_{\Omega}=-1/\left(t*g^{(q)}_{\Omega}\!\left(\bx\right)\right)$, respectively. The equality constraint \eqref{eq:optimization_problem_3_3} is included in the Lagrangian function as $h \left(\bx\right)=\bm{a}^T\bx-b$, with $\bm{a}=[1,1,0]^T$, and $b=p_{t,N}$.
\subsubsection{Responsive Primal-Dual IPM with Modified KKT system}
In this subsection, we propose a responsive primal-dual IPM with modified KKT conditions, presented in \textbf{Algorithm \ref{algo:PD_IPM}}, which aims to maximize the energy differential of minimum UE energy among all UEs. This procedure focuses the network energy in PA period $q$ towards a particular UE $k$ which has the lowest energy level at the end previous PA period $(q\!-\!1)$ \textbf{(Algorithm \ref{algo:PD_IPM}}, Line \ref{eq:algo_k_identification}). IPM provides an adequate solution for the second-order Lagrangian function presented in \eqref{eq:optimization_problem_4},  as it allows the search point to traverse the interior of the feasible conic region iteratively.

At a particular search point, the modified KKT conditions are necessary to ensure optimality of the IPM solution. These conditions for our problem in \eqref{eq:optimization_problem_4} are:
\begin{subequations} \label{eq:modified_kkt}
\begin{align}
     \nabla \fx + \left[\mathcal{J}\bgx\right]^T \bu + \bm{a}^T \bl   &=0, \\
     \bm{a}^T\bx-b&=0\\
    \mu_{\upsilon} \nabla g _{\upsilon}\!\left(\bx\right)+1/t&=0,\\
     \mu_{\Omega} \nabla g_{\Omega}\!\left(\bx\right)+1/t & = 0,\\
     \bgx\leq \bm{0}, \, \,\mu_{\upsilon} \geq 0, \,\, \mu_{\Omega} &\geq 0,
\end{align}
\end{subequations}
where $\mathcal{J}$ represents the Jacobian operation. The residual vector $\bm{r}^{(q)}_t\left(\bx,\bu,\bl\right)$ of this modified KKT system can be formed as:
\begin{equation}
    \bm{r}^{(q)}_t\left(\bx,\bu,\bl\right)=\begin{pmatrix}
       \nabla \fx + \left[\mathcal{J}\bgx\right]^T \bu + \bm{a}^T \bl\\
       - \diag(\bu)\bgx-\dfrac{1}{t} \bm{1}\\
       \bm{a}^T\bx-b
    \end{pmatrix}.
\end{equation}

After satisfying these necessary conditions, the search direction is sought out with Newton's method. Here, we denote the current search point as $\by=\left(\bx,\bu,\bl\right)^T$ and the primal-dual search direction $\Delta\by=\left(\Delta\bx,\Delta\bu,\Delta\bl\right)^T$ which can be characterized by Newton's method \cite{boyd2004convex} as $\mathcal{J}\bm{r}_t\left(\by\right)\Delta\by=-\bm{r}_t\left(\by\right)$, which is given in \eqref{eq:newton_method}:
  \begin{align}~\label{eq:newton_method}
      &\begin{pmatrix}
        \nabla^2 \fx +  \bu^T \nabla^2 \bgx & \mathcal{J} \bgx^T & \bm{a}^T \\
        - \diag\left(\bu\right) \mathcal{J}\bgx & - \diag\left(\bgx\right) & 0 \\
        \bm{a} & 0 & 0
        \end{pmatrix}
        \nonumber\\
        &\hspace{4em}\times
        \begin{pmatrix}
        \Delta \bx \\
        \Delta \bu \\
        \Delta \bl
        \end{pmatrix}
        = - \begin{pmatrix}
        \bm{r}_{\text{dual}, t}\left(\by\right) \\
        \bm{r}_{\text{cent}, t}\left(\by\right) \\
        \bm{r}_{\text{pri}, t}\left(\by\right)
    \end{pmatrix},
  \end{align}  
where $\bm{r}_{\text{dual}, t}\left(\by\right)$,$ \bm{r}_{\text{cent}, t}\left(\by\right)$ and $     \bm{r}_{\text{pri}, t}\left(\by\right)$ represent the primal, dual and centrality (modified complementarity) residual components. Now, we introduce the duality gap parameter to evaluate the distance to the optimal point over each iteration. However, $\by$ is not necessarily feasible, except in the vicinity of IPM convergence. Thus, we define the surrogate duality gap $\hat{\eta}= -\left[\bgx\right]^T \bu \,\,\forall \,\bx \to \fx\leq 0 \land \bm{\lambda} \geq 0$, which would be the duality gap if the solution point is feasible with $\bm{r}_{\text{pri}, t}\left(\by\right)=0$ and $\bm{r}_{\text{dual}, t}\left(\by\right)=0$. The iterative reduction of this gap results in the tightening of barrier constraint $t=\rho \dfrac{m}{\hat{\eta}}$ with $\rho > 1$.
\begin{algorithm}[t] 
\caption{Responsive Primal-Dual IPM}\label{algo:PD_IPM}
\begin{algorithmic}[1]
\Procedure{PD}{$ f,P,\bar{k}, \bxo, \bgx, \bm{a}, b, \epsilon, \epsilon_{\text{feas}}, \rho, q, \beta_{LS}$}
    \State $\bar{k} = \arg\min_{k \in \KE} E^{(q-1)}_k$ \label{eq:algo_k_identification}
    \State $\by \gets \left(\bxo, \bu, \bl\right)$ with initial $\bu > 0$ 
    \Repeat
        \State $\hat{\eta} \gets -\left[\bgx\right]^T \bu$
        \State $t \gets \rho \frac{ m}{\hat{\eta}}$
        \State $\Delta \by \gets -\left[\mathcal{J} \bm{r}_t\left(\by\right)\right]^{-1} \bm{r}_t\left(\by\right)$
        \State $\alpha \gets \text{PD-LS}\left(\bm{r}_t, \bgx, \by, \Delta \by, p, \beta_{LS}\right)$
        \State $\by \gets \by + \alpha \Delta \by$
    \Until{$\Bigl\| \bm{r}_{\text{dual}, t}\left(\by\right) \Bigl\|_2 \leq \epsilon_{\text{feas}} \land \Bigl\| \bm{r}_{\text{pri}, t}\left(\by\right) \Bigl\|_2 \leq \epsilon_{\text{feas}} \land \hat{\eta} \leq \epsilon$}
    \State \Return $\bm{y}^{(q)}_{{\bar{k}},x,opt}$
\EndProcedure
\end{algorithmic}
\end{algorithm}
\setlength{\textfloatsep}{0.5cm}

Based on the search direction $\Delta\by$, we update the search point by using the modified backtracking line-search procedure given in \textbf{Algorithm \ref{algo:linesearch}}, based on the residual norm to satisfy $\bx \to \fx\leq 0 \land \bm{\lambda} \geq 0$. At first, the maximum positive step length $\alpha_{max}$ is calculated, and then, backtracked with step reduction parameter $q$ and minimum descent $\beta_{LS}$. Utilizing this modified line-search procedure, the search point $\by$ is updated iteratively until both primal and dual residual $\left(\bm{r}_{\text{pri}, t}\left(\by\right), \bm{r}_{\text{dual}, t}\left(\by\right)\right)$ conform to the feasibility threshold $\epsilon_{\text{feas}} (\sim 10^{-9})$ along with the surrogate duality gap $\hat{\rho}$. At this stage, the optimization process culminates at the optimal power coefficients $\opkl$, which will maximize the harvested energy during the period $q$ for minimum energy UE $k$ at the end of the PA period $(q-1)$.
\begin{algorithm}[t]
\caption{Modified Backtracking Line-search for PD-IPM} \label{algo:linesearch}
\begin{algorithmic}[1]
\Procedure{PD-LS}{$\bm{r}_t, \bgx, \by, \Delta \by, q, \beta_{LS}$}
    \State $\alpha_{\text{max}} \gets \min\left(\{1\} \cup \left\{ \frac{-\byu}{\Delta \byu} \mid \Delta \byu < 0 \right\}\right)$
    \State $\alpha \gets 0.99 \alpha_{\text{max}}$
    \While{$ {g}_{\nu}\left(\byx \!+\! \Delta \byx\right) \!\geq \! 0 \,\lor {g}_{\Omega}\left(\byx \!+\! \Delta \byx\right)\! \geq \! 0 \,\lor \Bigl\| \bm{r}_t\left(\by \!+\! \alpha \Delta \by\right) \Bigl\|_2 \!>\! \left(1 \!-\! \alpha \beta\right) \Bigl\| \bm{r}_t(\by) \Bigl\|_2$}
        \State $\alpha \gets \alpha q$
    \EndWhile
    \State \Return $\alpha$
\EndProcedure
\end{algorithmic}
\end{algorithm}
\setlength{\textfloatsep}{0.5cm}

\textbf{Complexity analysis:} Now, we discuss the complexity of \textbf{Algorithm~\ref{algo:PD_IPM}} in the context of self-concordance framework applicable to SOCP problems \cite{boyd2004convex}. The upper bound on the total number of required Newton steps per PA period is given as
\begin{equation}\label{eq:newton_step1}
    \mathrm{Q} = \left\lceil \frac{\log(m / (t_{0} \epsilon))}{\log \rho} \right\rceil \left( \frac{m(\rho - 1 - \log \rho)}{\gamma} + c \right),
\end{equation}
where $m=u+v$ is the number of constraints, with $u\!=\!KL^2$ conic and $\!v=\!KL$ inequality constraints in \eqref{eq:optimization_problem_3}; $t_{0}=m/\eta_{0}$ is the initial value of $t$, which is defined by the initial duality gap $\eta_{0}=f\left(\bxo\right)-\bu_{0}^T \bm{g}\left(\bxo\right)$; $\gamma\!=\! \alpha\beta_{LS}(1\!-\!2\alpha)^2/(20\!-\!8\alpha)$ and $c\!=\!\log_2\log_2(1/\epsilon)$ are the scalar parameters in \eqref{eq:newton_step1}.  The choice of $\rho\!=\!\!1\!+\!1/\!\sqrt{m}$ ensures that outer iterations increase with $m$ in a logarithmic way. We can re-write \eqref{eq:newton_step1} as
\begin{equation}\label{eq:newton_step2}
    \mathrm{Q} = \left\lceil \sqrt{m}\log_2\left(\frac{m}{ t_{0} \epsilon}\right) \right\rceil \left( \frac{1}{2\gamma} + c \right).
\end{equation}
The expression in \eqref{eq:newton_step2}  illustrates that the computational complexity scales with $\mathrm{Q}$ as $\mathcal{O} \big(\sqrt{m}\log_2(m/\epsilon)\big)$. However, each Newton step further involves factorization of Hessian matrix in \eqref{eq:newton_method}, which typically has complexity $\mathcal{O}(K^3L^3)$ related to the third order of total number of optimization variables in \eqref{eq:optimization_problem_3}. Therefore, the overall algorithm complexity can be given as $\mathcal{O}\Big(K^3L^3\sqrt{KL(L+1)}\log_2\big(KL(L+1)/\epsilon\big)\Big)$.
\vspace{-0.5em}
\section{Numerical Examples}
\label{sec:results}
\subsection{System Parameters}
In this section, we consider a CF-mMIMO network to quantify the performance of both proposed power allocation schemes and visualize the temporal evolution of Markov process for minimum energy level of device storage in the whole network. The most important system parameters of the proposed CF-mMIMO have been tabulated in Table \ref{tab:sys_para}, which are based on CF-mMIMO literature \cite{hien,Demir,Wang:JIOT:2020}. The network is operational over a square area with APs forming a uniform square array. Here, the UEs have been considered to be randomly distributed over the network area. Moreover, each UE has a battery capacity of $100$ mJ which is divided into $M=2000$ discrete states to notice the demonstration of Markov process. Furthermore, the total network power and the total number of available service antennas ($LN$) have been fixed for a judicious comparison between varying network configurations. The antenna elements are distributed among APs whose total number belongs to the set $L\in \{4,9,16,25,36\}$, while for the case $L=25$ APs, the division remainder antennas are allocated to some APs randomly to keep $LN$ constant.

\begin{table}[t]
    \caption{\small System Parameters \normalsize}
    \vspace{-1ex}
    \centering
    \footnotesize
    \begin{tabular}{|p{3.9cm}| >{\centering\arraybackslash}m{1.1cm}| >{\centering\arraybackslash}m{2cm}|}
        \hline
        \hspace{3.5em}\textbf{Parameter} & \textbf{Symbol} & \textbf{Value}  \\
        \hline
        Coverage area & & $100 \times 100 \, \text{m}^{2} $
        \\
        \hline
        Total network power & $p_{t,N}$ & $10$ W
        \\
        \hline
        Number of APs & $L$ & $\{4,9,16,25,36\}$ \\
        \hline
        Number of service antennas & $LN$ & $288$ \\
        \hline
        Number of network UEs & $K$ & $20$ \\
        \hline
        AP height & $h_{AP}$ & $15$ m
        \\
        \hline
        UE height & $h_{UE}$ & $1.65$ m
        \\
        \hline
        Carrier frequency & $f$ & $1.9$ GHz \\
        \hline
        System bandwidth & $B$ & $20$ MHz \\
        \hline
        Coherence interval & $\tau_{c}$ & $0.2$s \\
         \hline
        Pilot duration  out of $\tau_{c}$
        & $\tau_{p}$ & $0.1$ \\
         \hline
        EH phase duration out of $\tau_{c}$
        & $\tau_{h}$ & $0.3$ \\
         \hline
        DL data phase duration out of $\tau_{c}$ 
        & $\tau_{d}$ & $0.3$\\
         \hline
        UL data phase duration  out of $\tau_{c}$
        & $\tau_{u}$ & $0.3$\\
        \hline
        Number of samples in $\tau_c$ &  & $200$ \\
        \hline
        Pilot signal power  & $p_{p}$ & $\{1,3,5\}\mu$W \\
        \hline
        UL data power  & $p_{u}$ & $\{1,3,5\}\mu$W \\
        \hline
        UE battery capacity &  & $100$ mJ \\
        \hline
        Number of discrete states & $M$ & $2000$ \\
        \hline
        Optimization phase time & $q$ & $2$s \\
        \hline
    \end{tabular}
    \vspace{-1em}
    \label{tab:sys_para}
\end{table}

\subsection{Large-Scale Fading Model}
The large-scale fading parameter is modeled as $\zekl=10^{-(PL_{kl}+\Psi_{kl})/10}$, encompassing both path loss $PL_{kl}$ and log-normal shadowing effects $\Psi_{kl}$. The shadowing factor can be represented as $\Psi_{kl}=\sigma_{\Psi} \tilde{\Psi}_{kl}$ with standard deviation $\sigma_{\Psi}$ and $\tilde{\Psi}_{kl}\!\sim\!\mathcal{N}(0, 1)$. A three-slope model is employed to characterize the path loss $PL_{kl}$, given in \cite{hien}, i.e., 
\begin{equation}
\hspace{-1em}PL_{kl}\! =\! 
\begin{cases} 
   -LP\!-\!35 \,{\log}_{10}(d_{kl}) &  d_{kl}\!>\! d_1, \\
   -LP\!-\!15 \,{\log}_{10}(d_{1})\!-\!20 \,{\log}_{10}(d_{kl}) &  d_0\!< \!d_{kl} \!< \!d_1, \\
   -LP\!-\!15 \,{\log}_{10}(d_{1})\!-\!20 \,{\log}_{10}(d_{0}) &  d_{kl}\! <\! d_0,\nonumber \\
\end{cases}
\end{equation}
where $d_0=10$ m, $d_1=50$ m \cite{zhao}, while
\begin{align}
    LP&\triangleq \,46.3+33.9 \log_{10}(f)\!-\!13.82 \log_{10}(h_{AP})\!\nonumber\\
    &-(1.1 \log_{10} (f)-0.7) h_{UE} +(1.56 \log_{10}(f)-0.8).
\end{align}
\vspace{-2em}
\subsection{Results and Discussion}
We first consider the PDF of the harvested energy $f(\EEk)$  in Fig. \ref{fig:PDF}. For this analysis, we have used the uniform power control which assigns equal power coefficients as $\Omega_{kl,unopt}=p_{t,N}(\sul \sum_{i \in \mathcal{K}} \gamil)^{-1}$, referred to as full power control (FPC)~\cite{hien}. It can be noticed that the distribution of the simulated harvested energy (solid colored bars) follows very closely the Gamma distribution approximation (solid black line) given in \eqref{eq:gam_pdf} based on the statistical parameters in \eqref{eq:gam_stat}. The energy harvested at all individual UEs follows the same pattern, however, for clear representation, the results of the particular UE, which has the minimum energy level among all network UEs, are shown in the above-referred figure. It can be further deduced that the harvested energy for more distributed service antennas is considerably higher in comparison to the centralized AP configurations for constant number of service antennas ($LN=288$), as we can observe that the PDF of $\EEk$ for the case $L=9$ APs is spread over a higher energy region than $L=4$ APs case.  


\begin{figure}[t]
    \centering
\vspace{-3em}    
\includegraphics[width=0.9\columnwidth]{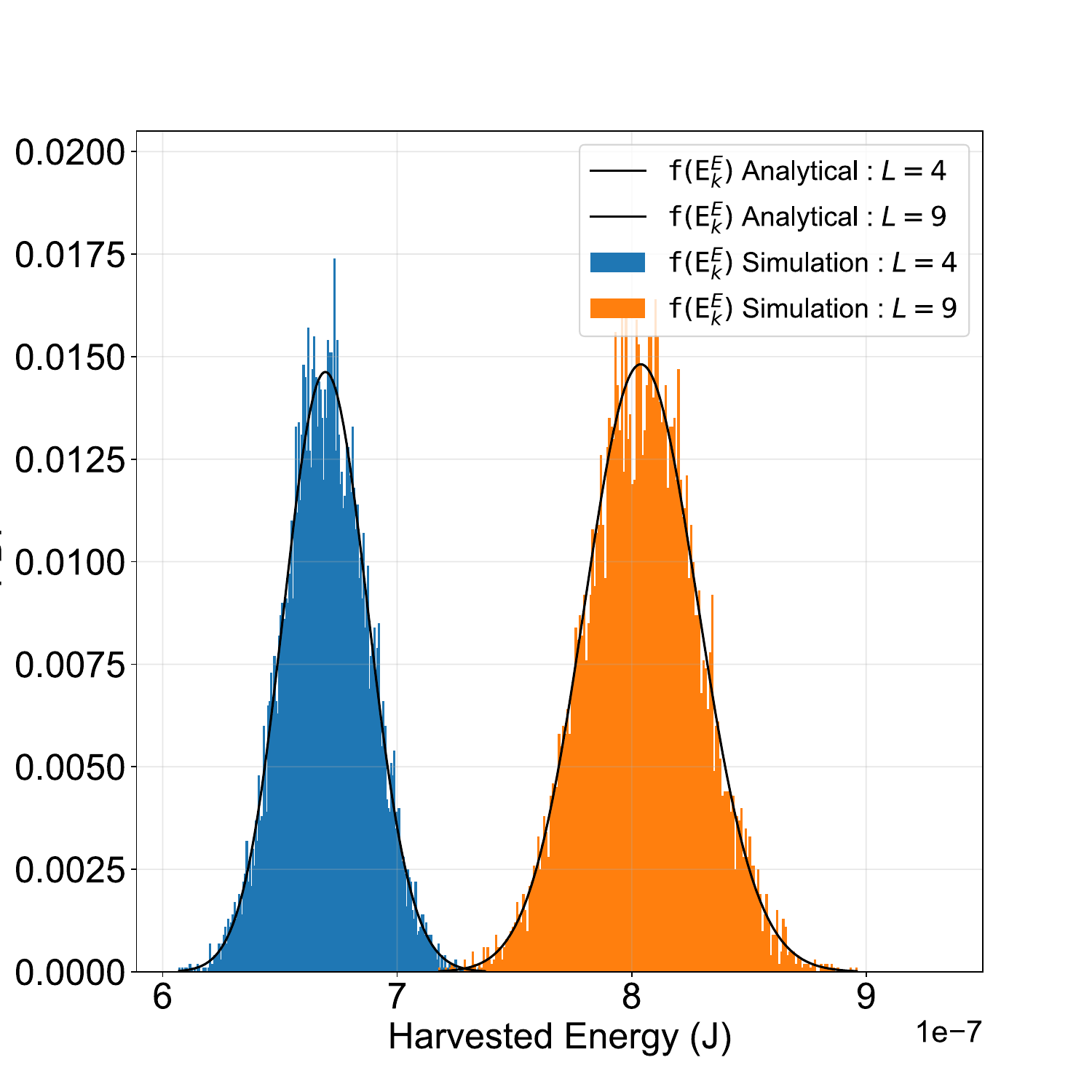}
    \vspace{0em}
    \caption{\small PDF of the harvested energy of the minimum energy UE for two different numbers of APs.\normalsize}
    \vspace{0em}
    \label{fig:PDF}
\end{figure}

The transition probabilities are critical indicators to evaluate the performance of Markov process based system model. Before discussing the insights related to these probabilities, it is pertinent to comment on the higher order transition probabilities ($P_{i,i\pm 2},P_{i,i\pm 3}, \hdots, P_{i,i\pm M}$). It is validated later in this section that the mean harvested energy $\ME\{\dEk\}$ for each UE within a coherence interval is in order $\sim 0.1-1\mu$J, which is negligible relative to the energy of a single discrete energy state $E_i$ ($50\mu$J). Hence, it is a rational assumption to consider that the higher order transition probabilities as zero and the analysis reduces to only three transition cases as given in \eqref{eq:transition_prob}, $P^k_{i,j} \in \{P_{i,i-1},P_{i,i},P_{i,i+1}\}$. These probabilities for energy state transition of UE device storage are reported in Table~\ref{tab:trans_probs}. A huge increase can be observed in the positive transition probabilities, along with similar scaling reduction in the negative transition probabilities as the number of APs increase from $L=4$ to $L=288$. This fact substantiates the efficacy of CF-mMIMO as the antennas are more distributed throughout the network spatial area. Moreover, the last case of $L=288$ single-antenna APs represents the maximum practical AP deployment, whose transition probabilities can be regarded as the upper-bound limits for the current number of service antennas ($LN=288$).

\begin{table}[t]
    \caption{\small Transition Probabilities of UE Energy States. \normalsize}
    \vspace{-1ex}
    \centering
    \footnotesize
    \begin{tabular}{|>{\centering\arraybackslash}m{1.6cm}|>{\centering\arraybackslash}m{1cm}|>{\centering\arraybackslash}m{1cm}|>{\centering\arraybackslash}m{1cm}|}
        \hline
        No. of APs & $P_{i,i}$ & $P_{i,i+1}$ & $P_{i,i-1}$ \\
        \hline
        \textbf{4} & 0.99710 & 0.00070 & 0.00220\\
        \hline
        \textbf{9} & 0.99845 & 0.00110 &0.00045\\
        \hline
        \textbf{16} & 0.99800 & 0.00160&0.00040\\
        \hline
        \textbf{25} & 0.99570 & 0.00400 &0.00030\\
        \hline
        \textbf{36} & 0.99480 & 0.00520 &0\\
        \hline
        \textbf{64 } & 0.98310 & 0.01690 & 0\\
        \hline
        \textbf{100}  & 0.97455 & 0.02545 & 0\\
        \hline
        \textbf{144} & 0.96090 & 0.03910 & 0\\
        \hline
        \textbf{288 } & 0.90350 & 0.09650 & 0
        \\\hline
    \end{tabular}
    \vspace{-1em}
    \label{tab:trans_probs}
    \vspace{0em}
\end{table}


\begin{figure}[t]
    \centering
        \begin{minipage}[b]{0.45\textwidth}
        \includegraphics[width=\columnwidth]{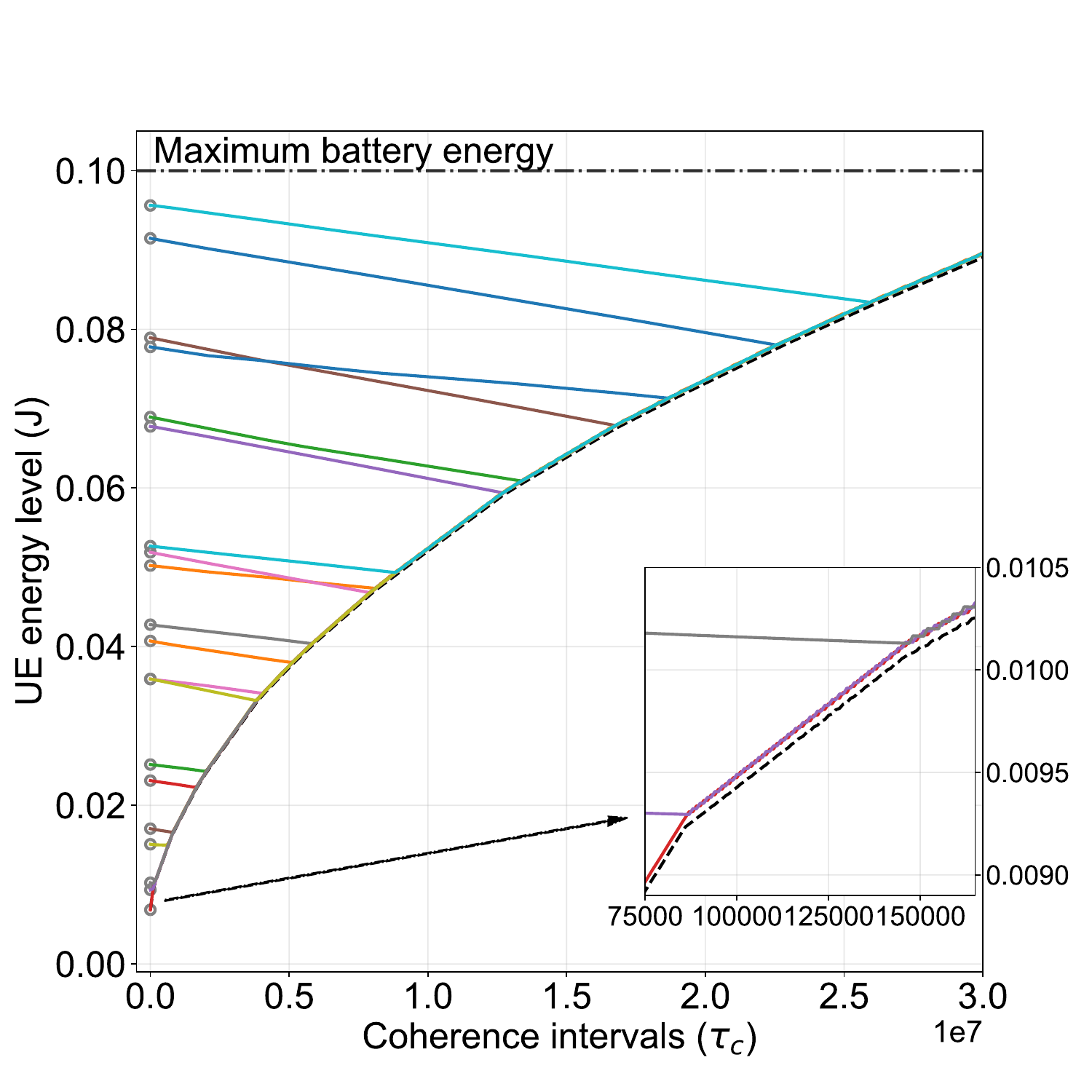}
        \vspace{-1em}
        \caption{\small Markov process progression of energy states $E_i$ for all UEs using the proposed PA schemes.\normalsize}
        \vspace{0.6em}
        \label{fig:markov_min}
    \end{minipage}
     \begin{minipage}[b]{0.45\textwidth}
     \vspace{-0.8em}
        \includegraphics[width=\columnwidth]{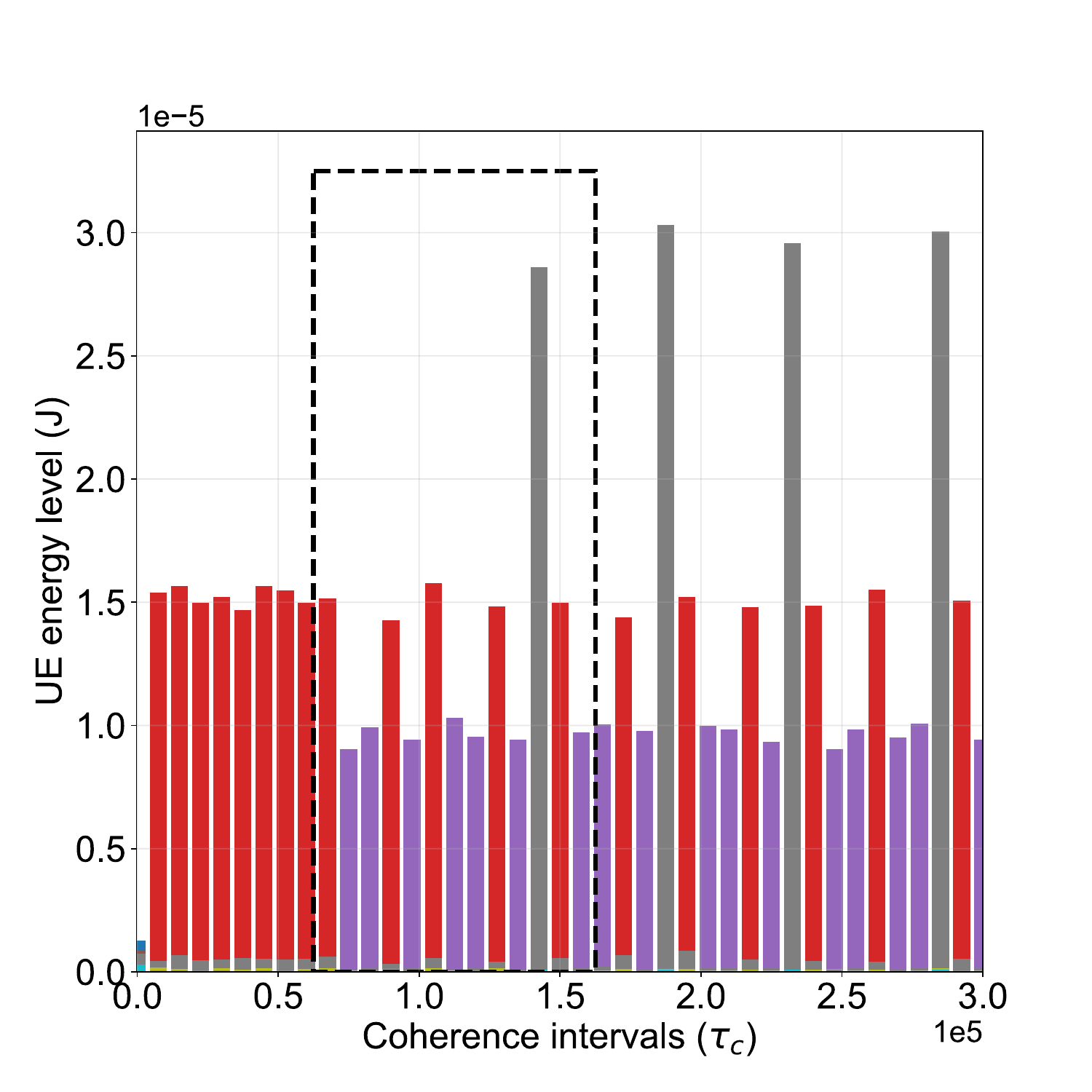}
        \vspace{-1em}
        \caption{\small Harvested energy $\EEk$ for initial $3\times 10^5 \tau_c$, whose portion is given in the inset axes of Fig.~\ref{fig:markov_min}.\normalsize}
        \label{fig:initial_EH_bar}
    \end{minipage}
\end{figure}

The temporal evolution of Markov process for the energy levels of all UEs in the network is shown in Fig. \ref{fig:markov_min}. This Monte-Carlo simulation spans $30$M number of coherence intervals $\tau_c$ ($\sim 17$ hours), with PA period $q$ carried out after every $2$ seconds considering the large-scale fading channel characteristics constant for this period. 
The proposed PA schemes are executed with the objective function based on the maximization of the minimum UE energy level among the network UEs, as given in \eqref{eq:optimization_problem_3}. The energy levels of all network UEs are shown in Fig. \ref{fig:markov_min} with solid colored lines for the optimization-based PA scheme and dashed black line for the heuristic PA scheme. Although most of the total power of the whole network during the EH phase is directed towards the minimum energy UE $k_1$, the majority of other UEs consume their stored energy for UL training and data transfer. Now, we focus on the inset axes of Fig. \ref{fig:markov_min}. It can be observed that the energy level of minimum UE $k_1$ (red solid line) crosses the UE $k_2$ with second minimum energy level (purple solid line). The EH optimization algorithm switches the directed power in between these two UEs until their energy level reaches the next lower energy UE $k_3$ (grey solid line). A more descriptive presentation of this concept is shown in Fig. \ref{fig:initial_EH_bar}, which shows the energy harvested in UEs during each PA period. Initially, a major portion of the power allocation is focused to UE $k_1$ with the minimum initial energy level (red bars). When this UE has harvested enough energy in battery storage and reaches the level of second minimum energy UE $k_2$, then the PA alternates between these two UEs (red and purple bars).


\begin{figure}[t]
    \centering
    \begin{minipage}[b]{0.45\textwidth}
        \includegraphics[width=\columnwidth]{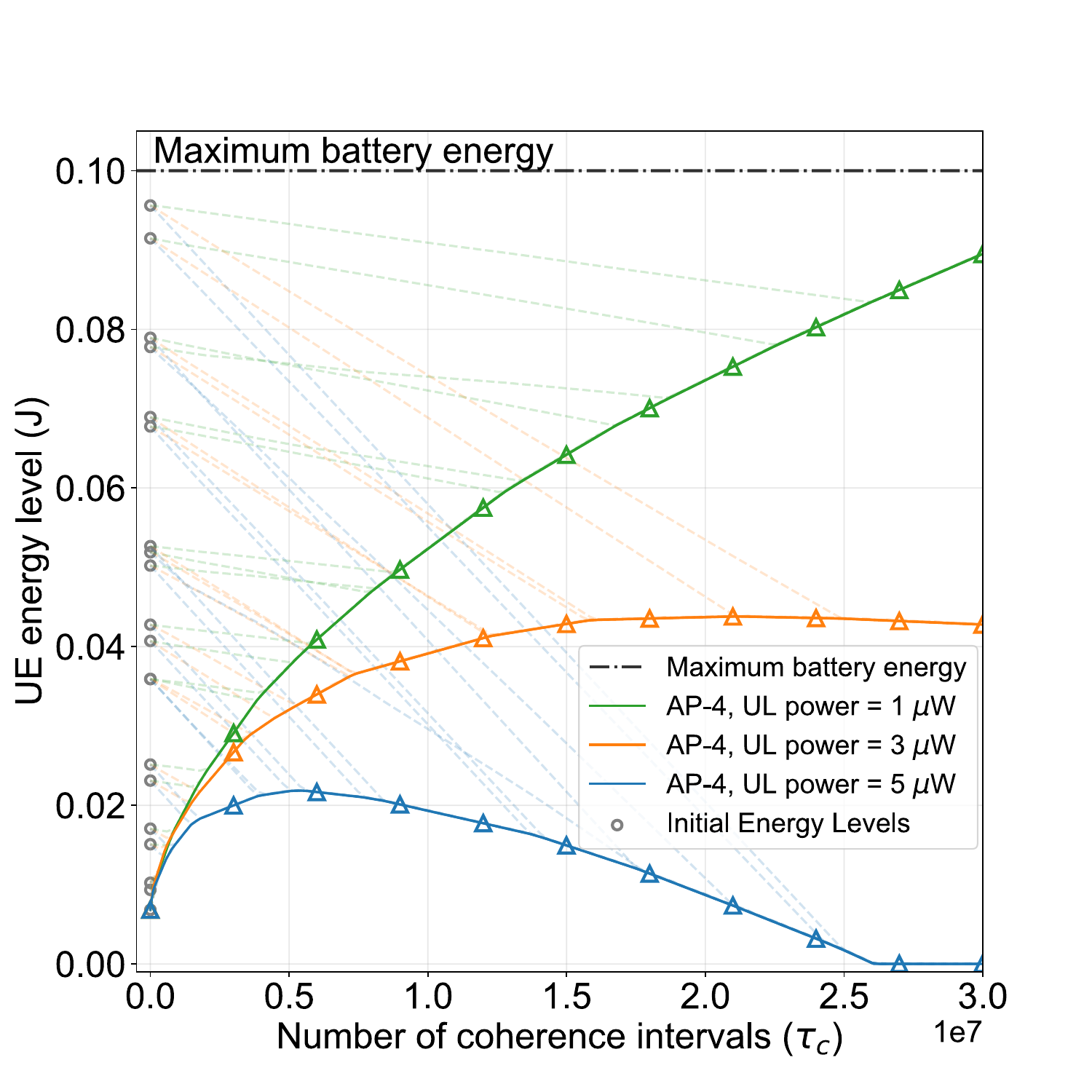}
        \vspace{-1em}
        \caption{\small Markov process evolution of a particular UE with minimum energy state for different $p_u$ levels.\normalsize}
        \vspace{-0.2em}
        \label{fig:steady_state}
    \end{minipage}
    \hspace{0.5em}
    \begin{minipage}[b]{0.45\textwidth}
        \includegraphics[width=\columnwidth]{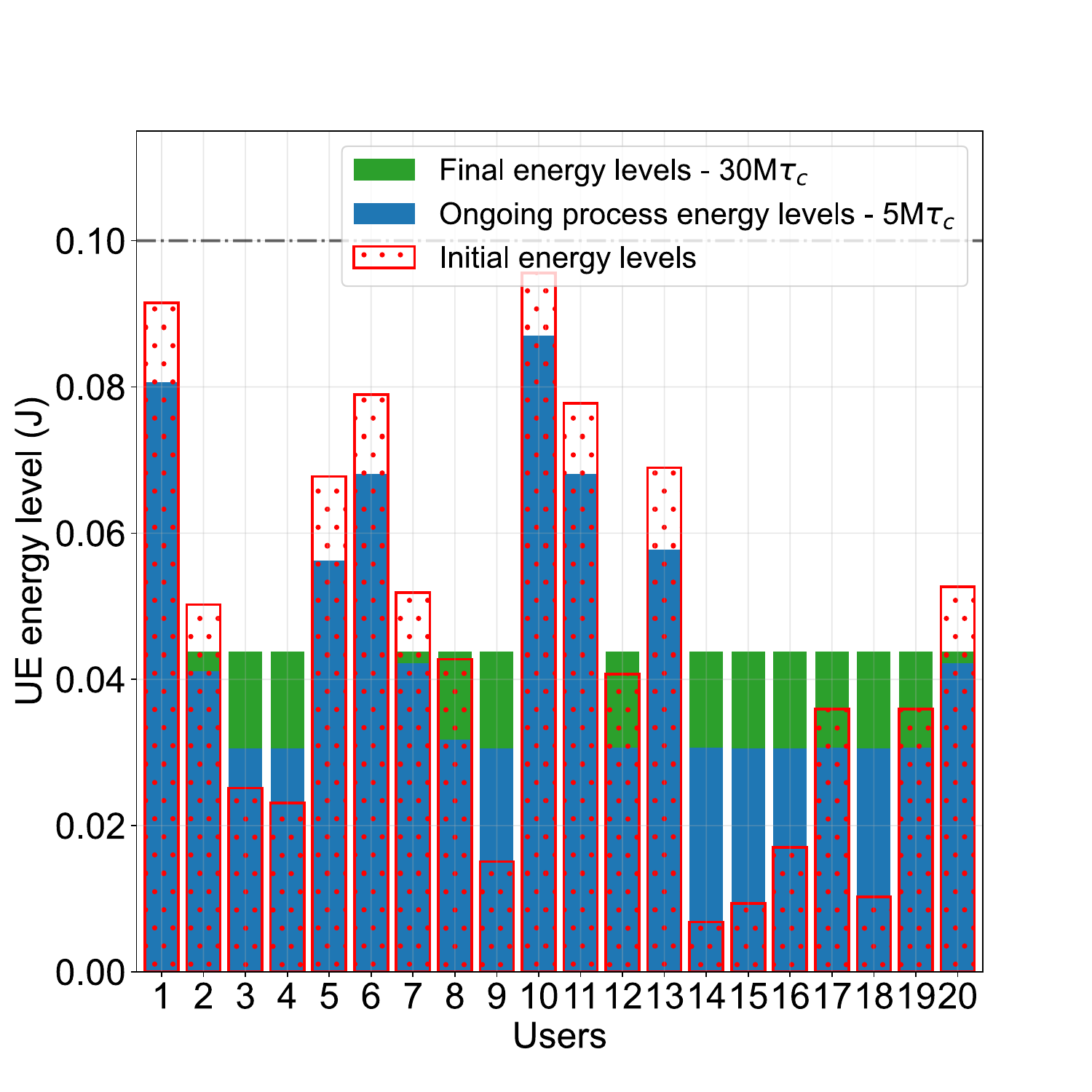}
        \vspace{-1em}
        \caption{\small Markov process evolution of energy state of all UEs with $p_u=3\mu$W and for different states: initial, $5$M $\tau_c$, and $30$M $\tau_c$.\normalsize}
        \vspace{-0.5em}
        \label{fig:bar_3uW}
    \end{minipage}
\end{figure}

The Markov process progression of UE energy level with minimum energy over $30$M coherence intervals is shown in Fig.~\ref{fig:steady_state} for three cases of UL data transfer/ pilot training power levels, $p_u \in\{ 1,3,5\}\mu$W, with AP configuration $L=4$ and total $LN=288$ service antennas. The choice of these levels for $p_u$ is reasonable, owing to the EH level in the proposed CF-mMIMO network which is completely reliant on WPT \cite{Demir,femenias}. It can be observed that the harvested energy for the case $p_u= 1\mu$W is significant enough that even the minimum energy UE approaches the maximum battery level as all the UEs have almost full battery level ($>90\%$ at $30$M $\tau_c$). In the $p_u= 5\mu$W case, EH is focused towards the minimum energy UE, while the other UEs consume their stored power for UL data and pilot training. After substantially long operational time, the EH process is not enough to support positive energy differential and after consuming all stored battery energy, the energy level of all UEs reaches zero. Now, we discuss the case of $p_u= 3\mu$W in which the final power levels after $30$ M coherence intervals saturate to $43.2\%$ level of the maximum battery storage. This equilibrium represents a balance between the energy harvested and the energy consumed within a single interval. We further analyze the temporal variation of energy level of all $K=20$ individual UEs in the network for this particular case in Fig.~\ref{fig:bar_3uW}. For presentation purposes, the red dotted bars show the initial UE energy levels at the start of the Markov process, while the blue bars show the energy levels after $5$M coherence intervals $\tau_c$ of network operation, while the green bars show the final energy level after $30$M $\tau_c$ have passed. It can be noticed that after 5M coherence intervals of network operation (blue bars), the UE energy levels with lower initial energy levels have improved significantly, whereas there is a slight drop in the energy levels of UEs with high initial energy. The reason behind this observation is that the power coefficients for EH have been optimized to favour the minimum energy UE during the progression of Markov process, while high initial energy UEs have very less harvested power during this period and have to consume their stored energy for UL data transmission and pilot training. After $30$M $\tau_c$, the final energy level (green bars) reaches the steady state at $43.2\%$ level.


\begin{figure}[t]
    \centering
    \vspace{-1em}
    \includegraphics[width=\columnwidth]{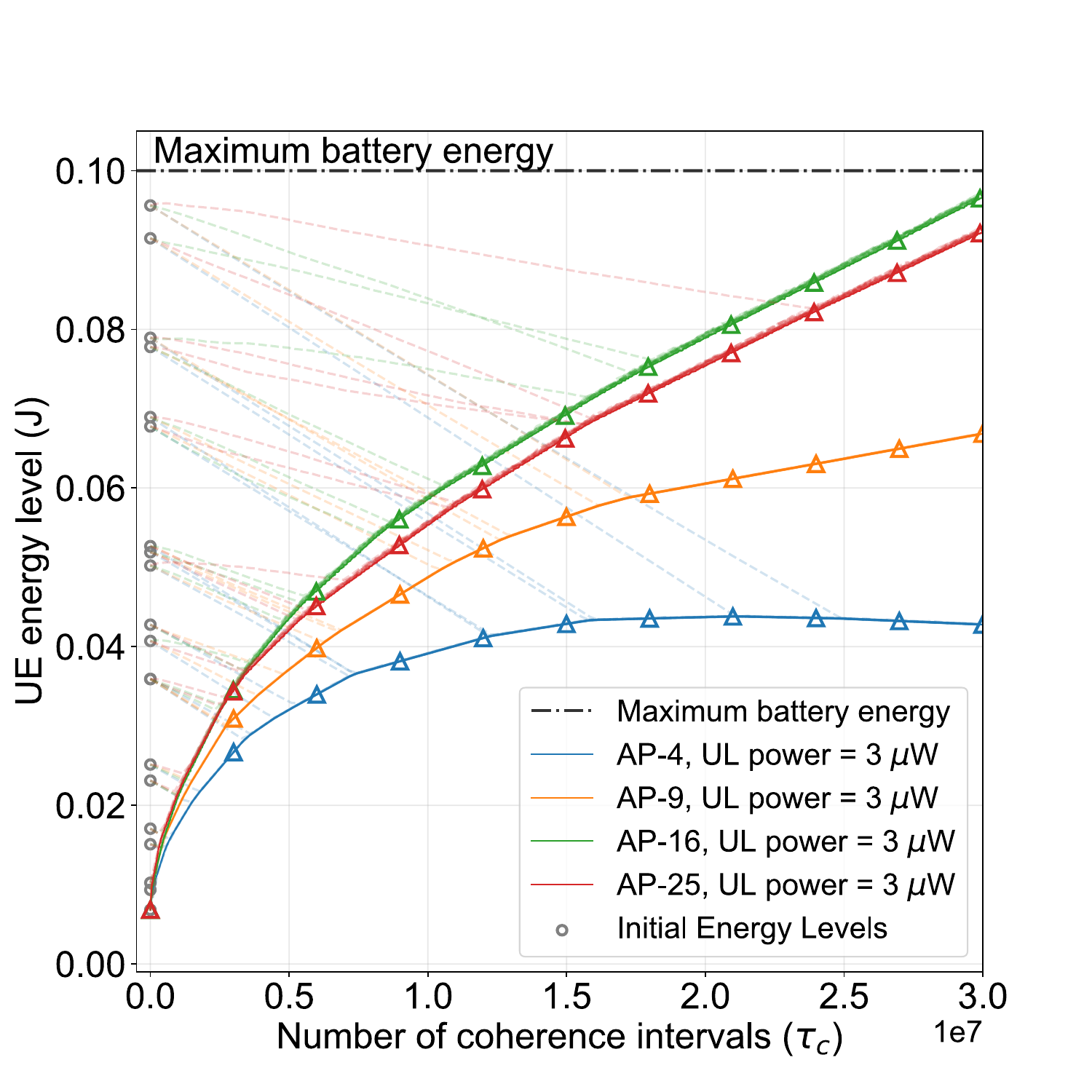}
        \vspace{-1em}
        \caption{\small Markov process evolution of the UE with minimum energy state for different AP configurations at $p_u=3\mu$W.\normalsize}
        \label{fig:EH_vs_APs}
\end{figure}

\begin{figure}[t]
    \centering
    \vspace{-1em}
    \includegraphics[width=\columnwidth]{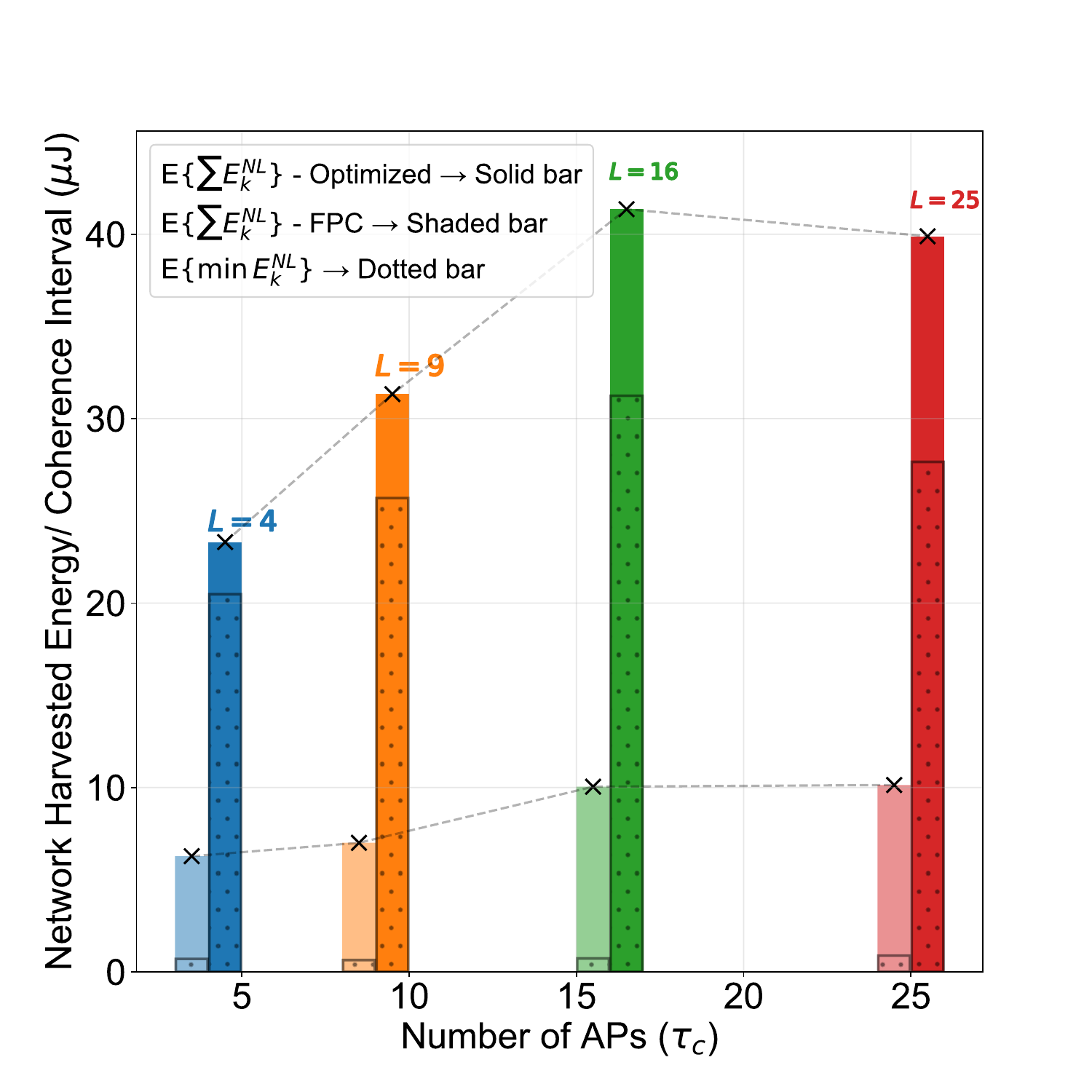}
        \vspace{-1em}
        \caption{\small Mean of the network sum-EH ($\ME\{\sum \EEk\}$) and EH for minimum energy UE ($\ME\{\min \EEk\}$) for different AP configurations.\normalsize}
        \label{fig:HE_vs_AP}
\end{figure}

The relationship between spatial distribution of constant service antennas in CF-mMIMO and the amount of the harvested energy is investigated in Fig. \ref{fig:EH_vs_APs}, in connection with the proposed optimization algorithm for EH maximization for minimum energy UE. The results presented here show the battery energy level of minimum energy UE  over the period of $30$ M coherence intervals with Markov process based iterative  optimization, while consuming $p_u=3\mu$W for UL data transfer/ pilot training. As the number of APs increases from $L=4$ to $16$, a remarkable improvement can be observed in the minimum energy level at final state, with $55.1\%$ increase for $L=9$ ($67\%$ level) and $125\%$ increase for $L=16$ ($97.2\%$ level) in comparison to the $L=4$ case ($43.2\%$ level). On the other hand, 4.3\% reduction has been noticed for $L=25$ ($93\%$ level) in comparison to the $L=16$ case. The reason for this reduction can be attributed to the smaller number of available antennas at each AP, resulting in diminished beamforming gain and accordingly received power at particular UEs. 

Next, we evaluate the EH efficiency of the proposed optimization algorithm in comparison with FPC allocation for fair comparison~\cite{hien, Zhang:IoT:2022, Wenchao}. The results presented in Fig.~\ref{fig:HE_vs_AP} assess the mean of total network sum-EH ($\ME\{\sum \EEk\}$ - solid bars) and the mean EH for minimum energy UE ($\ME\{\min \EEk\}$ - dotted bars) with the proposed optimized EH displayed as filled bars relative to FPC based EH as shaded bars, over different AP configurations. It can be clearly visualized that the mean optimized EH for the network sum case is around four times more than the mean of FPC based EH for all AP setups, while for the minimum energy UE, this factor is at least 30 times. Moreover, it is evident that the major portion ($\sim 70-85\%$) of optimized EH is focused on the minimum energy UE (dotted vs solid bars) against the network sum within each PA period, which exhibits the successful manifestation of the proposed optimization algorithm anchored on the cardinal objective of maximization of the minimum energy UE.

As a final point, we present the performance comparison for two configurations $L=16$ and $36$ APs in Fig. \ref{fig:Heuristic}. We evaluate the two proposed PA schemes against another heuristic scheme (EPA) which allocates the total network power from all APs to the minimum energy UE $\bar{k}$ without factoring in the channel conditions. As a final point, we present the performance comparison for two configurations $L=16$ and $36$ APs in Fig. \ref{fig:Heuristic}. We evaluate the two proposed PA schemes against another heuristic scheme which allocates the total network power from all APs to the minimum energy UE $\bar{k}$ without factoring in the channel conditions. Mathematically speaking,
\vspace{0.4em}
\begin{align}\label{eq:heuristic1}
    \Omega^{(q)}_{kl,H^{*}}=
    \begin{cases}
        & p_{t,N}/L \quad \forall \quad k=\bar{k}\\ 
        & 0 \quad \quad \quad \,\, \forall \quad  k\neq \bar{k}.
    \end{cases}
\end{align}
\begin{figure}[t]
    \centering
    \vspace{-1em}
    \includegraphics[width=0.9\columnwidth]{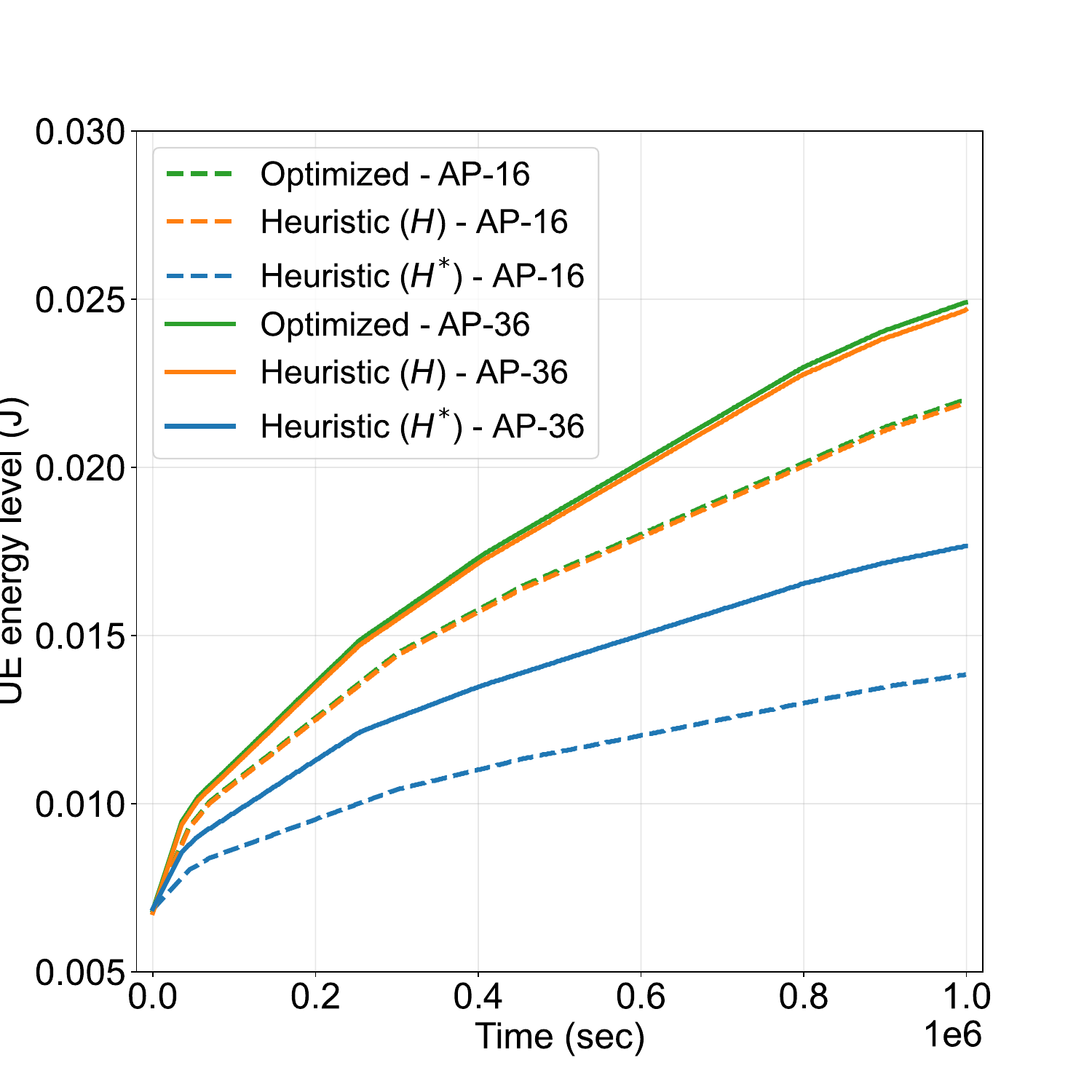}
        \vspace{-1em}
        \caption{\small Markov process progression of minimum energy states $E_i$ for all UEs with different PA schemes.\normalsize}
        \label{fig:Heuristic}
\end{figure}
\vspace{-0.4em}

It can be observed that the proposed optimized power allocation serves as the upper bound for the EH performance in CF-mMIMO networks. Note that the performance of the heuristic CCPA given in \eqref{eq:heuristic2} is comparable to the optimized case with a mean error less than $1\%$ for both AP scenarios. This impressive performance has been achieved with significantly lower computational time compared to the optimized algorithm, as given in Table \ref{tab:comp_time} which covers all scenarios from $L=4$ to $36$ APs.  Furthermore, the EH performance of the heuristic EPA ($H^*$) is $62.9 \%$ and $70.9 \%$ of the optimized power allocation scheme for the scenarios $L=16$ and $L=36$ APs, respectively. From Table \ref{tab:comp_time}, we can notice that this heuristic scheme requires even less computational time than the other power allocation schemes.

\begin{table}[t]
    \vspace{-1ex}
    \caption{\small Computational time (seconds) for various power allocation schemes. \normalsize}
    \vspace{-1ex}
    \centering
    \footnotesize
    \begin{tabular}{| >{\centering\arraybackslash}m{1.0cm} |>{\centering\arraybackslash}m{2.0cm}|>{\centering\arraybackslash}m{2.0cm}|>{\centering\arraybackslash}m{2.0cm}|}
        \hline
        \textbf{No. of APs} & \textbf{Optimized} & 
        \textbf{Heuristic EPA} & \textbf{Heuristic CCPA } \\
        \hline
        \textbf{4} & $0.6937$ s & $1.99\, \mu$s & $57.77\, \mu$s\\
        \hline
        \textbf{9} & $4.1718$ s & $7.15\, \mu$s & $102.39\, \mu$s\\
        \hline
        \textbf{16} & $17.8430$ s & $12.71\, \mu$s & $165.75\, \mu$s\\
        \hline
        \textbf{25} & $31.8838$ s & $20.64\, \mu$s & $278.22\, \mu$s\\
        \hline
        \textbf{36} & $97.6355$ s & $35.23\, \mu$s & $429.65\, \mu$s
        \\\hline
    \end{tabular}
    \vspace{-1em}
    \label{tab:comp_time}
\end{table}

\section{Conclusion}
\label{sec:conclusion}
We have provided a robust framework for optimized EH CF-mMIMO networks using a Markov process evolution. By integrating discrete energy state transitions with comprehensive statistical modeling and power allocation techniques, the proposed methods enhanced the energy levels of UE devices across the network significantly, specifically the minimum energy UE in each PA period. The simulation results corroborated the efficacy of the Gamma distribution in characterizing the harvested energy, demonstrating a four-fold improvement in the minimum UE energy levels through both dynamic PA approaches involving optimized PA using the PD-IPM and channel estimates based heuristic PA. This advancement not only addresses the critical challenge of sustainable energy supply in IoT networks but also underscores the potential of CF-mMIMO systems in overcoming path loss limitations and enhancing spatial resource diversification. The effect of this factor, combined with the impact of the proposed PA frameworks, has been observed as a substantial enhancement in the network's minimum energy levels with increasingly dispersed AP configurations up to a certain level. Beyond this point, the effective beamforming gain diminishes due to the reduced number of service antennas available at each AP. 

Future research can build on these findings by exploring further refinements in the PA approaches over Markov process, such as leveraging machine learning techniques for maximizing the energy efficiency along with improved resilience to channel impairments. The applicability of these research themes of Markov process based CF-mMIMO system to other wireless communication paradigms, like ultra-dense networks and heterogeneous networks, where the high nodal density and varying network conditions entail unique challenges and opportunities for optimized and dynamic energy management.

\appendices
\section{Useful Lemmas}~\label{apx:lemms}
\vspace{-2em}
\begin{Lemma}~\label{lemma:expxhAx}~\cite[Eq. (15.14)]{Kay}: 
Let ${\qu}$ be a complex $n \times 1$ random vector with mean $\boldsymbol{\mu}$ and covariance matrix $\boldsymbol{\Sigma}$ and let $\qB$ be an $n\times n$ positive definite Hermitian matrix. Then, we have
$\Ex\{{\qu}^H \qB  {\qu} \} = \boldsymbol{\mu}^H \qB \boldsymbol{\mu} + \trace (\qA \boldsymbol{\Sigma}).$
\end{Lemma}

\begin{Lemma}~\label{lemma:expxhAxsq}~\cite[Lemma 2]{emil}: 
Let ${\qu}$ be a complex $n \times 1$ random vector with mean $\boldsymbol{\mu}$ and covariance matrix $\boldsymbol{\Sigma}$ and let $\qB$ be an $n\times n$ positive definite Hermitian matrix. Then, we have
\begin{align}
 \ME\{\lvert\qu^H \qB \qu\rvert^2\}=\lvert \mathrm{tr}(\pmb{\Sigma}\qB )\rvert^2+\mathrm{tr}(\pmb{\Sigma}\qB\pmb{\Sigma}\qB^H).   
\end{align}
\end{Lemma}

\begin{Lemma}~\label{lemma:moments}~\cite[Eq. (2.10)]{verdu}: 
For a central complex Wishart matrix \(\qW= \qA \qA^H \sim \mathcal{W}_m(n,\qI)\), the following identity holds:
\begin{equation*}
    \centering
    \ME\{\mathrm{det} \qW^k\}=\prod^{m-1}_{l=0} \dfrac{\Gamma(n-l+k)}{\Gamma(n-l)}.
\end{equation*}
For the special case, when $\qA$ is replaced by vector $\tgkl^H$, we have $\qW= \tgkl^H\tgkl$, and $m=1, n=N$. Then, we have the following result
\begin{equation}
    \ME\{\lVert\tgkl\rVert^{2k}\}=\dfrac{\Gamma(n+k)}{\Gamma(n)}.
\end{equation}
\end{Lemma}

\section{Proof of Proposition~\ref{Lemma:EIK}}
\label{Proof:Lemma:EIK}
%

By taking the expectation of~\eqref{eq:harvested_power} over the small-scale fading, we get
\begin{align}
    \ME\big\{\!\IEk\!\big\}
    \!= \!\! \sul \slp\suik \!
    \kmrti \!\kmrtip\!\sqrt{\Omega_{il}\Omega_{il'}}  \ME\big\{\gklt \hgilc \hgiltp \gklcp \big\}.
\label{eq:harvested_power_final}
\end{align}

To simplify the solution, we de-synthesize this composite harvested power into two parts: The coherent signal from the same AP for the correlated UEs ($l'= l$) and the non-coherent signal ($l'\neq l$).

\subsubsection{Compute the coherent case $l'= l$}: The terms for this case can be summed as:  
\begin{align}~\label{eq:mean:case1}
        \ME \big\{\gklt \hgilc \hgilt \gklc\big\}
         & =\, \, \bgklt\bgilc\bgilt\bgklc+ \ME\{\bgklt\thgilc\thgilt\bgklc\}
        \nonumber\\
        &\hspace{-5em}+\!
        \ME\{\thgklt\bgilc\bgilt\thgklc\}
        \!+\!\ME\{\thgklt\thgilc\thgilt\thgklc\}\!+\! \ME\{\teklt\bgilc\bgilt\teklc\}\nonumber\\
        &\hspace{-5em}+\!\ME\{\teklt\thgilc\thgilt\teklc\}\!+\! \ME\{\bgklt\bgilc\thgilt\thgklc\} \!+\! \ME\{\thgklt\thgilc\bgilt\bgklc\}
        \nonumber\\
        &\hspace{-5.5em} = \! N\big(N \bkilsq \oiklsq \!+\! \bklq \gamil\!+\! \bilq \gamkl +\aiksq (N+1) \gamklsq \nonumber\\
        &\hspace{-5em}+\!(\bilq+\gamil) (\bkl-\gamkl) \!+\! 2 N \gamkl \aik \bkil \oikl \big),
\end{align}
where we have used Lemma \ref{lemma:moments} to calculate $\ME\{\thgklt\thgilc\thgilt\thgklc\}= \aiksq \ME\{\thgklt\thgklc\thgklt\thgklc\}=\aiksq N(N+1) \gamklsq$.

\subsubsection{Compute the non-coherent case $l'\neq l$}: Following the independence of the channels, we get 
\begin{align}~\label{eq:mean:case2}
    \ME\left\{\gklt \hgilc \hgiltp \gklcp \right\} &= \ME\left\{ \gklt \hgilc\right\}\ME\big \{\hgiltp \gklcp\big\}\nonumber\\
    &\hspace{-6em}= N^2(\bkil \oikl+ \aik \gamkl)(\bkilp \oiklp+\aik \gamklp).
\end{align}
To this end, by substituting~\eqref{eq:mean:case1} and~\eqref{eq:mean:case2} into~\eqref{eq:harvested_power_final}, the desired result in~\eqref{eq:mean:final} can be obtained. 

Now, we focus on the calculation of the variance of harvested power, we can write as
\begin{align}
    \MV\Big\{\IEk\Big\}
    = & \sul \suik \left(\kmrti\right)^4 \Omega^2_{il} \MV\left\{\gklt \hgil^* \hgilt \gklc \right\}+\nonumber\\
    & \sul \slpL \suik \left(\kmrti \kmrtip\right)^2 {\Omega_{il}\Omega_{il'}}\nonumber\\
    &\times\MV \big\{ \gklt \hgil^* \hgiltp \gklcp\big\}.
\label{eq:harvested_power_var_final}
\end{align}

In order to segregate the zero-variance terms, we define $R_{k,i,l}= \gklt \hgil^*$, which is used to represent $\gklt \hgil^* \hgilt \gklc = R_{k,i,l}R^H_{k,i,l}$. This parameter can be expressed as
\begin{align}~\label{eq:Rkil}
    \begin{split}
       R_{k,i,l}
       &= \bgkl^T\bgilc+\aik\bgkl^T\thgklc+\sqrt{\bkl} \tgklt\bgilc+\aik \sqrt{\bkl} \tgklt\thgklc.
    \end{split}
\end{align}

\subsubsection{Compute the coherent term $\MV\left\{\gklt \hgil^* \hgilt \gklc \right\}$} \label{subsec:coherent}
Firstly, we will evaluate variance for coherent term of harvested power in \eqref{eq:harvested_power_var_final}. Noticing that $\MV\{\bgklt\bgilc\bgilt\bgklc\}=0$, we have
\vspace{-0.3em}
\begin{align}
\label{eq:rtimesrh}
    \MV\{R_{k,i,l}R^H_{k,i,l}\}=& \, 2 \, \Big(\MV\big\{\aik\bgklt\thgklc\bgilt\bgklc\big\}+  
    \nonumber\\
    &\hspace{-7em}
    \MV\big\{\sqrt{\bkl} \tgklt\bgilc\bgilt\bgklc\big\}+ \MV\big\{\aik \sqrt{\bkl} \tgklt\thgklc\bgilt\bgklc\big\}+\nonumber\\
    &\hspace{-7em}
    \!\MV\big\{\aik\sqrt{\bkl} \tgklt\bgilc\thgklt\bgklc\big\}\!+
    \!\MV\big\{\aiksq \sqrt{\bkl} \tgklt\thgklc\thgklt\bgklc\big\}+\nonumber\\
    &\hspace{-7em}
    \!\MV\big\{\aik \bkl \tgklt\thgklc \bgilt \tgklc\big\}\Big)+
    \!\MV\big\{\aiksq\bgklt\thgklc\thgklt\bgklc\big\}+ \nonumber\\
    &\hspace{-7em}
    \MV\big\{\bkl \tgklt\bgilc \bgilt \tgklc\big\}+\MV\big\{\aiksq \bkl \tgklt\thgklc \thgklt\tgklc\big\}.
\end{align}

Now, we will proceed to compute each individual term in~\eqref{eq:rtimesrh}. The first term can be derived
\vspace{-0.3em}
\begin{align}
    \MV\{\aik\bgklt\thgklc\bgilt\bgklc\}&= \aiksq \bklq \bkilsq  \ME\{ \lvert\bhklt\thgklc\bhilt\bhklc\rvert^2\}\nonumber\\
    &\hspace{0em}
    =\aiksq \bklq \bkilsq  N \oiklsq  \ME\{ \bhklt\thgklc\thgklt\bhklc\}\nonumber \\
    &= N^3 \aiksq \bklq \bkilsq \oiklsq \gamkl,
\end{align}
whereas, the second term in~\eqref{eq:rtimesrh} can be readily calculated as 
\vspace{-0.3em}
\begin{align}
    \MV\{\sqrt{\bkl} \tgklt\bgilc\bgil^T\bgklc\}= N^3 \bkl \bkilsq \oiklsq \bilq.
\end{align}

The third term in \eqref{eq:rtimesrh} can be obtained as
\vspace{-0.3em}
\begin{align}
    \MV\{\aik \sqrt{\bkl} \tgklt\thgklc\bgilt\bgklc\}&=\aiksq \bkl \bkilsq \ME\{\lvert\tgklt\thgklc\bhilt\bhklc\rvert^2\}\nonumber\\
    &\hspace{-3em}=\aiksq \bkl \bkilsq N \oiklsq \ME\{\tgklt\thgklc\thgklt\tgklc\}\nonumber\\
    &\hspace{-3em}= N^2 \aiksq \bkilsq \oiklsq \gamkl ( \bkl\!+\!N\gamkl),\!\!
\end{align}
where we have used the following result (from Lemma~\ref{lemma:moments})
\vspace{-0.3em}
\begin{align}
  \ME\{\tgklt\thgklc \thgklt\tgklc\}  &= \ME\{(\teklt+\thgklt)\thgklc \thgklt(\teklc+\thgklc)\} /\bkl
  \nonumber\\
  &
  = 
  (\ME\{\teklt\thgklc \thgklt\teklc\} 
  +
  \ME\{\Vert\thgkl\Vert^4\})/\bkl\nonumber\\
  &
  = N\gamkl (\bkl+N\gamkl)/\bkl.
\label{eq:var_t_16_2}
\end{align}
The fourth term in \eqref{eq:rtimesrh} can be computed as
\begin{align}
    \MV\{\aik\sqrt{\bkl} \tgklt\bgilc\thgklt\bgklc\}
    &= \aiksq \bkl \bkilsq \ME\{\lvert\tgklt\bhil^*\thgklt\bhkl^*\rvert^2\}    \nonumber\\
    &= \aiksq \bkl \bkilsq N^2 \gamkl.
\end{align}
Now, we focus on the derivation of the fifth term in~\eqref{eq:rtimesrh}, which can be obtained as
\begin{align}
     \!\!\MV\{\aiksq \sqrt{\bkl} \tgklt\thgklc\thgklt\bgklc\} &= \aiktt \bkl \bklq\ME\{\lvert\tgklt\thgklc\thgklt\bhklc\rvert^2\}\nonumber\\
     &\hspace{-8em}
     =  \aiktt  \bklq     N(N+1)\gamma_{kl}^2 \big((N+1)\gamkl + \beta_{kl} \big).
\label{eq:var_t_8}
\end{align}
In~\eqref{eq:var_t_8}, we have used the following derivation (by applying Lemma~\ref{lemma:moments}) 
\begin{align}
\ME\{\lvert\tgklt\thgklc\thgklt\bhklc\rvert^2&=
\ME\{\tgklt\thgklc\thgklt\thgklc\thgklt\tgklc\}
\nonumber\\
&\hspace{-6em}
=
(\ME\{\Vert \thgkl \Vert^6\} + \ME\{\teklt\ME\{\thgklc\thgklt\thgklc\thgklt\}\teklc\})/\bkl
\nonumber\\
&\hspace{-6em}
= (N(N+1)\gamma_{kl}^2 ((N+1)\gamkl + \beta_{kl} ))/\bkl.
\end{align}
The sixth term in~\eqref{eq:rtimesrh} can be derived as
\begin{align}\label{eq:var_t_12}
   \MV\{\aik \bkl \tgklt\thgklc \bgil^T \tgklc\}
   &=\aiksq \bklsq \bilq \ME\{\lvert\tgklt\thgklc \bhil^T \tgklc\rvert^2\}\nonumber\\
   &\hspace{-10em}
   = \aiksq \bilq N(N+1)\gamkl \Big(  
   (N-1) \bkl\gamkl + \bklsq + 2 \gamklsq\Big).
\end{align}
The expectation in \eqref{eq:var_t_12} can be obtained as
\begin{align}
    \begin{split}
        \ME\{\lvert\tgklt\thgklc \bhil^T \tgklc\rvert^2\}
        &=  
        \ME\{\tgklt\thgklc \tgklt \tgklc\thgklt \tgklc\} 
        \nonumber\\
    &\hspace{-7em} =
        \ME\{(\thgklt\!+\!\teklt)\thgklc (\thgklt\!+\!\teklt) (\thgklc\!+\!\teklc)\thgklt (\thgklc\!+\!\teklc)\}/\bklsq\nonumber\\
        &\hspace{-7em}=  
        \Big(\ME\{\thgklt\thgklc \thgklt \thgklc\thgklt \thgklc\}+\ME\{\teklt\thgklc \teklt \teklc\thgklt \teklc\}+\nonumber\\
        &\hspace{-6em}\ME\{\teklt\thgklc \thgklt \thgklc\thgklt \teklc\}+\ME\{\thgklt\thgklc \teklt \teklc\thgklt \thgklc\}\Big)/\bklsq\nonumber\\
        &\hspace{-7em}= \Big(  
        N(N+1)(N+2) \gamklcb+ N (N+1)  \gamkl (\bkl-\gamkl)^2+\nonumber\\
        &\hspace{-6em}N (N\!+\!1) \gamklsq (\bkl\!-\!\gamkl)+ N^2 (N\!+\!1) \gamklsq (\bkl\!-\!\gamkl)\Big)/\bklsq\nonumber\\
        &\hspace{-7em}= N(N+1)\gamkl \Big(  
        (N-1) \bkl\gamkl + \bklsq + 2      \gamklsq\Big)/\bklsq.
    \end{split}
\end{align}
The seventh term in \eqref{eq:rtimesrh} can be evaluated as
\begin{align}\label{eq:var_t_6}
   \MV\{\aiksq\bgkl^T\thgklc\thgklt\bgklc\} &= \aiktt \bklqsq\Big(\ME\{\lvert\bhkl^T\thgklc\thgklt\bhkl^*\rvert^2\}\!-\nonumber\\
        &\hspace{1em}\!\big(\ME\{\lvert\bhkl^T\thgklc\thgklt\bhkl^*\rvert\}\big)^2 \Big)\nonumber\\
        &\hspace{0em}
        =\aiktt \bklqsq \gamklsq N^2. 
\end{align}
Here,  we have applied Lemma~\ref{lemma:expxhAxsq} from Appendix~\ref{apx:lemms} to derive $\ME\{\lvert\bhkl^T\thgklc\thgklt\bhkl^*\rvert^2\}\!=\! \ME\{\lvert\thgklt\bhklc\bhklt\thgklc\rvert^2\}$, with $\qu =\thgklc$, $\qB= \bhklc\bhklt$, $\pmb{\Sigma}= \gamkl \In$. By substitution, we get the first term in \eqref{eq:var_t_6},
 \begin{align}
          \ME\{\lvert\thgklt\bhkl^*\bhkl^T\thgklc\rvert^2\}&=\!  \gamklsq  \left \lvert \mathrm{tr}\!\left(\bhkl^*\bhkl^T\right) \right \rvert^2 \!\!+\! \gamklsq \mathrm{tr} \left(\bhkl^*\bhkl^T\bhkl^*\bhkl^T\right) \nonumber\\
        &\hspace{0em}
          =2\gamklsq N^2.
 \end{align}
 The second term in \eqref{eq:var_t_6} can be derived as,
 $\ME\{\lvert\bhkl^T\thgklc\thgklt\bhkl^*\rvert\}=\ME\{\lvert\thgklt\bhkl^*\bhkl^T\thgklc\rvert\}=\ME\{\thgklt\thgklc\}=N \gamkl$. 

The eight term in \eqref{eq:rtimesrh} is calculated as
\begin{align} \label{eq:var_t_11}
    \MV\{\bkl \tgklt\bgilc \bgilt \tgklc\} &=\bklsq \bilqsq \Big(\ME\{\lvert\tgklt\bhilc \bhilt \tgklc\rvert^2\}-\nonumber\\
        &\hspace{1em}
        \left(\ME\{\tgklt\bhilc \bhilt \tgklc\}\right)^2\Big)\nonumber\\
        &=\bklsq \bilqsq N^2.
\end{align}
To derive $\ME\{\lvert\tgklt\bhilc\bhilt\tgklc\rvert^2\}$, we apply Lemma~\ref{lemma:expxhAxsq} from Appendix~\ref{apx:lemms} with $\qu =\tgklc$, $\qB= \bhilc\bhilt$, $\pmb{\Sigma}= \In$. By substitution, we get the first term in \eqref{eq:var_t_11}, $\ME\{\lvert\tgklt\bhilc \bhilt \tgklc\rvert^2\}=\left \lvert \mathrm{tr}\!\left(\bhilc\bhilt\right) \right \rvert^2\!\!+\!  \mathrm{tr}\!\left(\bhilc\bhilt\bhilc\bhilt\right)=2N^2$. The second term in \eqref{eq:var_t_11} can be derived as $\ME\{\tgklt\bhilc \bhilt \tgklc\}= \ME\{\tgklt \tgklc\}= N$. 

The last term in \eqref{eq:rtimesrh} can be computed as:
\begin{align}\label{eq:var_t_16}
    \MV\{\aiksq \bkl \tgklt\thgklc \thgklt\tgklc\}=&\aiktt \bkl^2 \Big(\ME\{\lvert\tgklt\thgklc \thgklt\tgklc\rvert^2\}\nonumber\\
    &-\big(\ME\{\tgklt\thgklc \thgklt\tgklc\}\big)^2\Big).        
\end{align}
The first expectation term in \eqref{eq:var_t_16} can be derived as:
\begin{align}
\ME\{\lvert\tgklt\thgklc \thgklt\tgklc\rvert^2\}
        &=\ME\{(\thgklt\!+\!\teklt)\thgklc \thgklt(\thgklc\!+\!\teklc) (\thgklt\!+\!\teklt)\nonumber\\
        &\times\thgklc \thgklt(\thgklc\!+\!\teklc)\}/\bklsq\nonumber\\
        &\hspace{-8em} = \Big( \ME\{\lVert\thgkl\rVert^8\}
        +
        \ME\{ \thgklt\thgklc\thgklt\thgklc \teklt \thgklc\thgklt \teklc \}
        \!\nonumber\\
        &\hspace{-7.5em}+\!\ME\{\thgklt\thgklc\thgklt\teklc\thgklt\thgklc\thgklt\teklc \}\!+\!\ME\{\thgklt\thgklc\thgklt\teklc\teklt\thgklc\thgklt\thgklc \}\!
        \!\nonumber\\
        &\hspace{-7.5em}+\!\ME\{\teklt\thgklc\thgklt\thgklc\thgklt\thgklc\thgklt\teklc\}\!+\!\ME\{\teklt\thgklc\thgklt\thgklc\teklt\thgklc\thgklt\thgklc\}
        \!\!\nonumber\\
        &\hspace{-7.5em}+\!\ME\{\teklt\thgklc\thgklt\teklc\thgklt\thgklc\thgklt\thgklc\}\!+\!\ME\{\lvert\teklt\thgklc\thgklt\teklc\rvert^2\}\!\Big)/\bklsq.\nonumber\\
        &\hspace{-8em}= N(N\!\!+\!\! 1)(N\!\!+\!\!2)(N\!\!+\!\!3) \gamkltt\!+\!4 N(N\!\!+\!\!1)(N\!\!+\!\!2)(\bkl\!\!-\!\!\gamkl) \gamklcb \!\nonumber\\
        &\hspace{-7em}+\! (\bkl\!\!-\!\!\gamkl)^2\gamklsq N(2N\!\!+\!\!1).
\end{align}

The term $\ME\{\lvert\teklt\thgklc\thgklt\teklc\rvert^2\}$ has been derived using Lemma \ref{lemma:expxhAxsq} with $\qu= \teklc$, $\qB= \thgklc\thgklt$ and $\pmb{\Sigma}= \ME\{\teklt\teklc\}= (\bkl-\gamkl)\In$. We have considered the independence of channel estimate ($\thgkl$) and estimation error ($\tekl$), while using this expression:
\vspace{-0.3em}
\begin{align}
    \begin{split}
        \ME\{\lvert\qu^H\qB \qu\rvert^2\big \vert \qB\}&\!=\!(\bkl\!-\!\gamkl)^2\left(\lvert \mathrm{tr}(\qB) \rvert^2 \!+\!  \mathrm{tr}(\qB \qB^H)\right).
    \end{split}
\end{align}
Now, the expectation of this conditional expectation will provide the desired result as:
\vspace{-0.3em}
\begin{align}
    \begin{split}
        \ME\{\lvert\qu^H\qB \qu\rvert^2\big \} & = \ME\{\ME\{\lvert\qu^H\qB \qu\rvert^2\big \vert \qB\}\}=(\bkl-\gamkl)^2\nonumber\\
        &\quad\times\Big(\lvert \mathrm{tr}(\ME\{\thgklc\thgklt\}) \rvert^2 +  \mathrm{tr}(\ME\{\thgklc\thgklt\thgklc\thgklt\})\Big)\\
        &=(\bkl-\gamkl)^2\big(N^2 \gamkl^2 +  N(N+1) \gamklsq\big)\nonumber\\
        &=(\bkl-\gamkl)^2\gamklsq N(2N+1).\\
    \end{split}
\end{align}
The second expectation term in \eqref{eq:var_t_16} has been already derived in \eqref{eq:var_t_16_2}. By substitution, we get
\vspace{-0.3em}
\begin{align}
    \MV\{\aiksq \bkl \tgklt\thgklc \thgklt\tgklc\}=&\aiktt  \Big( N(N\!+\! 1)(N\!+\!2)(N\!+\!3) \gamkltt \nonumber \\
    &\hspace{-6.5em}+\!4 N(N\!+\!1)(N\!+\!2)(\bkl\!-\!\gamkl) \gamklcb \!+\! (\bkl\!-\!\gamkl)^2\nonumber\\
    &\hspace{-6.5em} \times \gamklsq N(2N\!+\!1)\!-\!N^2\gamklsq (\bkl\!+\!N\gamkl)^2\Big).
\end{align}
\subsubsection{Compute non-coherent term $\MV \big\{ \gklt \hgil^* \hgiltp \gklcp\big\}$}
Now, we consider the non-coherent part of the variance of the harvested power in~\eqref{eq:harvested_power_var_final}. To this end, we have
\vspace{-0.3em}
\begin{align}\label{eq:rtimesrh'}
    &\MV\{R_{k,i,l}R^H_{k,i,l'}\}\! =\MV\{\aik\bgklt\thgklc\bgiltp\bgklcp\}\nonumber\\
    &+\!\MV\{\sqrt{\bkl} \tgklt\bgilc\bgiltp\bgklcp\}+\MV\{\aik \sqrt{\bkl} \tgklt\thgklc\bgiltp\bgklcp\}\nonumber\\
    &\hspace{0em}+\!\MV\{\aik\sqrt{\bkl} \tgklt\bgilc\thgkltp\bgklcp\}\!+\!\MV\{\aiksq \sqrt{\bkl} \tgklt\thgklc\thgkltp\bgklcp\}\!\!\nonumber\\
    &\hspace{0em}+\!\MV\{\aik\sqrt{\bklp}\bgklt\thgklc \bgiltp\tgklcp\}\!+\!\MV\{\sqrt{\bkl\bklp} \tgklt\bgilc \bgiltp\tgklcp\}\nonumber\\
    &\hspace{0em}+\!\MV\{\sqrt{\bklp}\bgklt\bgilc \bgiltp\tgklcp\}\!+\!\MV\{\aik \sqrt{\bkl\bklp} \tgklt\thgklc \bgiltp\tgklcp\} \nonumber\\
    &\hspace{0em}+\!\MV\{\aik \sqrt{\bklp}\bgklt\bgilc \thgkltp\tgklcp\}\!\!+\!\!\MV\{\aiksq\sqrt{\bklp}\bgklt\thgklc  \thgkltp\tgklcp\!\}\nonumber\\
    &+\!\MV\{\aik\bgklt\bgilc\thgkltp\bgklcp\}+\!\MV\{\aik \sqrt{\bkl\bklp} \tgklt\bgilc \thgkltp\tgklcp\}\nonumber\\
    &+\!\MV\{\aiksq\bgklt\thgklc\thgkltp\bgklcp\}+\MV\{\aiksq \sqrt{\bkl\bklp} \tgklt\thgklc  \thgkltp\tgklcp\},
\end{align} 
while $\MV\{\bgklt\bgilc\bgiltp\bgklcp\}=0$. Now, we shall evaluate the individual terms in \eqref{eq:rtimesrh'} sequentially. Following the approaches presented in Section \ref{subsec:coherent}, we can derive the non-coherent part of the variance expression as~\eqref{eq:vareq} at the top of the next page.
\begin{figure*}
\begin{align}\label{eq:vareq}
    \MV \big\{ \gklt \hgil^* \hgiltp \gklcp\big\}&=N^3 \aiksq \bklq \bkilpsq \oiklsqp \gamkl+N^3 \oiklsqp \bkl \bilq \bkilpsq+N^2 \aiksq \bkilpsq \oiklsqp \gamkl (\bkl+N\gamkl)\nonumber\\
    &+ N^3 \aiksq  \oiklsq \bkilsq \bklqp  \gamklp+N^2\aiktt \bklq \bklqp \gamkl \gamklp +N^2\aiksq \bkl \bilq \bklqp\gamklp +N^2 \aiksq \bklp \bklq\bilqp  \gamkl\nonumber\\
    &+ N^2 \aiktt  \bklqp \gamklp \gamkl (\bkl+N\gamkl)+N^3 \oiklsq \bklp \bilqp \bkilsq +N^2\aiksq \bklp \bilqp \gamkl(\beta_{kl}+N\gamkl)\nonumber\\
    &+N^2\bkl \bklp \bilq \bilqp +N^2 \oiklsq \aiksq \bkilsq \gamklp (\bklp+N \gamklp)+\aiktt \bklq N^2 \gamkl \gamklp (\bklp+N\gamklp)\nonumber\\
    &+\aiksq \bkl  \bilq N^2 \gamklp (\bklp\!+\!N\gamklp)\!+\!\aiktt N^2\gamkl\gamklp\Big(\bkl\bklp\!+\!N(\gamkl\bklp\!+\!\bkl \gamklp)\Big).
\end{align}    
\hrulefill
\vspace{-1.5em}
\end{figure*}
Here, the term $\MV\{\aiksq \sqrt{\bkl} \tgklt\thgklc\thgkltp\bgklcp\}$ in \eqref{eq:rtimesrh'} is derived using \eqref{eq:var_t_16_2} as
\begin{align}
    \MV\{\aiksq \sqrt{\bkl} \tgklt\thgklc\thgkltp\bgklcp\} &\!=\! \aiktt \bkl \bklqp \ME\{\lvert\tgklt\thgklc\thgkltp\bhklcp\rvert^2\}\nonumber\\
    &\hspace{-3em}=N  \aiktt \bkl \bklqp \gamklp \ME\{\tgklt\thgklc\thgklt\tgklc\}\nonumber\\
    &\hspace{-3em}=N^2 \aiktt \bklqp \gamklp \gamkl (\bkl\!+\!N\gamkl).\!
\end{align}
A similar technique has been used for the derivation of $\MV\{\aik \sqrt{\bkl}\sqrt{\bklp} \tgklt\thgklc \bgiltp\tgklcp\}$ using \eqref{eq:var_t_16_2}. Note that, the term $\MV\{\aiksq \sqrt{\bkl}\sqrt{\bklp} \tgklt\thgklc \thgkltp\tgklcp\}$ can be derived as 
\begin{align}
        \MV\{\aiksq \sqrt{\bkl}\sqrt{\bklp} \tgklt\thgklc \thgkltp\tgklcp\}&=\aiktt \bkl\bklp \nonumber\\
        &\hspace{-12em}\times \Big(\ME\{\lvert\tgklt\thgklc \thgkltp\tgklcp\rvert^2\}-\big(\ME\{\tgklt\thgklc \thgkltp\tgklcp\}\big)^2\Big)\nonumber\\
        &\hspace{-13em}=\!\aiktt N^2\gamkl\gamklp\Big(\bkl\bklp\!+\!N(\gamkl\bklp\!+\!\bkl \gamklp)\Big),\!\!
\end{align}
where $\ME\{\lvert\tgklt\thgklc \thgkltp\tgklcp\rvert^2\}\!=\!\ME\{\tgklt\thgklc \ME\{\thgkltp\tgklcp\tgkltp\thgklcp\!\}\thgklt\tgklc\}\!\!=\! N^2\gamkl\gamklp (\bkl+N\gamkl)(\bklp+N\gamklp)/(\bkl \bklp)$, and $\ME\{\tgklt\thgklc \thgkltp\tgklcp\}=\ME\{\tgklt\thgklc\}\ME\{\thgkltp\tgklcp\}= N^2 \gamkl \gamklp/(\sqrt{\bkl \bklp})$.

\bibliographystyle{IEEEtran}
\bibliography{main}
\vspace{-3em}
\begin{IEEEbiography}[{\includegraphics[width=1in,height=1.25in,clip,keepaspectratio]
{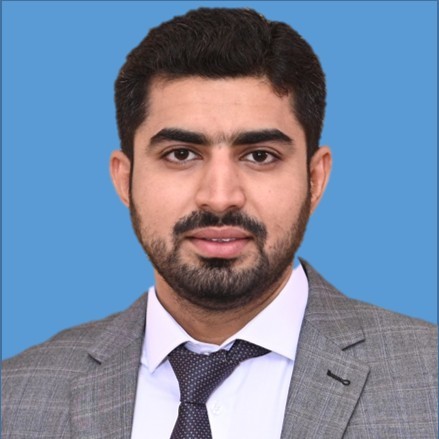}}]
{Muhammad Zeeshan Mumtaz} (Graduate Student Member, IEEE) received the B.E. degree in Avionics engineering from National University of Sciences \& Technology (NUST), Pakistan, in 2014, and the M.S. degree in Avionics Engineering from Air University, Pakistan, in 2021. He is currently pursuing the Ph.D. degree in Electrical \& Electronics Engineering at Queen's University Belfast, U.K. From 2021 to 2023, he was a Lecturer of Data Communications \& Networking at College of Aeronautical Engineering, NUST Pakistan. His research interests include cell-free massive MIMO systems, NOMA communication systems, MIMO radars, autonomous modulation classification and the application of deep learning techniques for contemporary communication challenges.
\end{IEEEbiography}

\begin{IEEEbiography}[{\includegraphics[width=1in,height=1.25in,clip,keepaspectratio]
{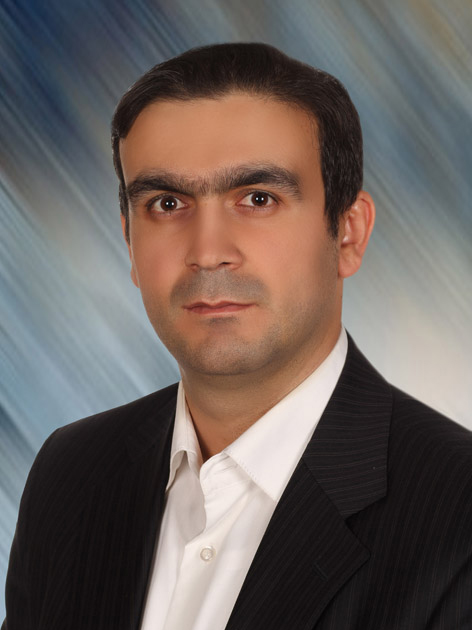}}]
{Mohammadali Mohammadi} (Senior Member, IEEE) is currently a Lecturer at the Centre for Wireless Innovation (CWI), Queen’s University Belfast, U.K. He previously held the position of Research Fellow at CWI from 2021 to 2024. His research interests include signal processing for wireless communications, cell-free massive MIMO, integrated sensing and communications, and reconfigurable intelligent surfaces. He has published more than 80 research papers in accredited international peer reviewed journals and conferences in the area of wireless communication and has co-authored two invited book chapters. He serves as an Associate Editor for IEEE Communications Letters and IEEE Open Journal of the Communications Society. He was a recipient of the Exemplary Reviewer Award for IEEE Transactions on Communications in 2020 and 2022, and IEEE Communications Letters in 2023. He has been a member of Technical Program Committees for many IEEE conferences, such as ICC, GLOBECOM, and VTC.
\end{IEEEbiography}

\begin{IEEEbiography}[{\includegraphics[width=1in,height=1.25in,clip,keepaspectratio]{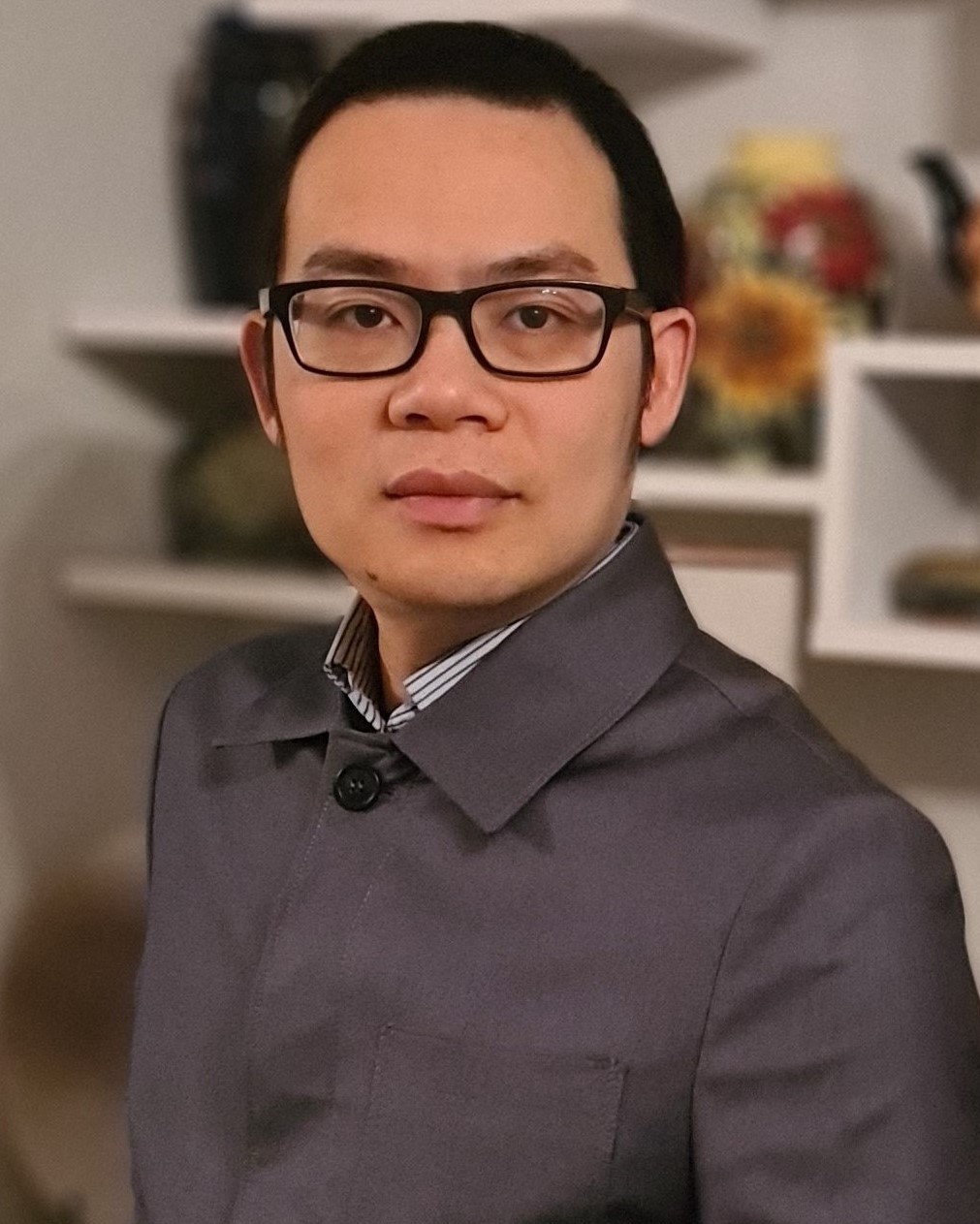}}]
{Hien Quoc Ngo} (Fellow, IEEE)  is currently a Reader with Queen's University Belfast, U.K. His main research interests include massive MIMO systems, cell-free massive MIMO, reconfigurable intelligent surfaces, physical layer security, and cooperative communications. He has co-authored many research papers in wireless communications and co-authored the Cambridge University Press textbook \emph{Fundamentals of Massive MIMO} (2016).

He received the IEEE ComSoc Stephen O. Rice Prize in 2015, the IEEE ComSoc Leonard G. Abraham Prize in 2017, the Best Ph.D. Award from EURASIP in 2018, and the IEEE CTTC Early Achievement Award in 2023. He also received the IEEE Sweden VT-COM-IT Joint Chapter Best Student Journal Paper Award in 2015. He was awarded the UKRI Future Leaders Fellowship in 2019. He serves as the Editor for the IEEE Transactions on Wireless Communications, IEEE Transactions on Communications, the Digital Signal Processing, and the Physical Communication (Elsevier). He was an Editor of the IEEE Wireless Communications Letters, a Guest Editor of IET Communications, and a Guest Editor of IEEE ACCESS in 2017.
\end{IEEEbiography}

\vspace{-2em}
\begin{IEEEbiography}[{\includegraphics[width=1in,height=1.35in,clip,keepaspectratio]{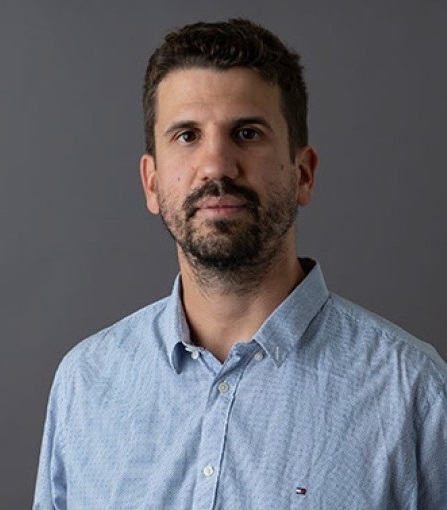}}]
{Michail Matthaiou}(Fellow, IEEE) obtained his Ph.D. degree from the University of Edinburgh, U.K. in 2008. 
He is currently a Professor of Communications Engineering and Signal Processing and Deputy Director of the Centre for Wireless Innovation (CWI) at Queen’s University Belfast, U.K. He is also an Eminent Scholar at the Kyung Hee University, Republic of Korea. He has held research/faculty positions at Munich University of Technology (TUM), Germany and Chalmers University of Technology, Sweden. His research interests span signal processing for wireless communications, beyond massive MIMO, reflecting intelligent surfaces, mm-wave/THz systems and AI-empowered communications.

Dr. Matthaiou and his coauthors received the IEEE Communications Society (ComSoc) Leonard G. Abraham Prize in 2017. He currently holds the ERC Consolidator Grant BEATRICE (2021-2026) focused on the interface between information and electromagnetic theories. To date, he has received the prestigious 2023 Argo Network Innovation Award, the 2019 EURASIP Early Career Award and the 2018/2019 Royal Academy of Engineering/The Leverhulme Trust Senior Research Fellowship. His team was also the Grand Winner of the 2019 Mobile World Congress Challenge. He was the recipient of the 2011 IEEE ComSoc Best Young Researcher Award for the Europe, Middle East and Africa Region and a co-recipient of the 2006 IEEE Communications Chapter Project Prize for the best M.Sc. dissertation in the area of communications. He has co-authored papers that received best paper awards at the 2018 IEEE WCSP and 2014 IEEE ICC. In 2014, he received the Research Fund for International Young Scientists from the National Natural Science Foundation of China. He is currently the Editor-in-Chief of Elsevier Physical Communication, a Senior Editor for \textsc{IEEE Wireless Communications Letters} and \textsc{IEEE Signal Processing Magazine}, an Area Editor for \textsc{IEEE Transactions on Communications} and Editor-in-Large for \textsc{IEEE Open Journal of the Communications Society}. 
\end{IEEEbiography}

\end{document}